\documentclass[11pt]{article}
\usepackage[colorlinks=true,linkcolor=blue!60!black,citecolor=blue!75!black]{hyperref}%
\usepackage{booktabs}
\usepackage{tabularx}
\usepackage{array}

\usepackage{amsmath,bm,amssymb,amsthm}
\usepackage{algorithm,algpseudocode}

\algnewcommand{\ParFor}{\textbf{parfor}}
\algnewcommand{\EndParFor}{\textbf{end parfor}}

\newtheorem{corollary}{Corollary}
\newtheorem{proposition}{Proposition}

    \makeatletter
    \renewcommand\section{\@startsection {section}{1}{\z@}%
                                       {-3.5ex \@plus -1ex \@minus -.2ex}%
                                       {2.3ex \@plus.2ex}%
                                       {\normalfont\fontfamily{phv}\fontsize{14}{17}\bfseries}}
    \renewcommand\subsection{\@startsection{subsection}{2}{\z@}%
                                         {-3.25ex\@plus -1ex \@minus -.2ex}%
                                         {1.5ex \@plus .2ex}%
                                         {\normalfont\fontfamily{phv}\fontsize{12}{15}\bfseries}}
    \renewcommand\subsubsection{\@startsection{subsubsection}{3}{\z@}%
                                        {-3.25ex\@plus -1ex \@minus -.2ex}%
                                         {1.5ex \@plus .2ex}%
                                         {\normalfont\normalsize\fontfamily{phv}\fontsize{11}{14}\selectfont}}
    \makeatother
	
	\usepackage{graphicx}
	\usepackage{enumerate}
	\usepackage{xcolor}
	\usepackage{natbib} 
	\usepackage{url} 
	\usepackage{booktabs}
    \usepackage{caption}
    \usepackage{multirow}
    \usepackage{subfig}
    \usepackage{array}

\newcommand{\AlgComment}[1]{\hfill{\scriptsize\itshape \# #1}}

\title{
Post-Corrected Raw-Score Martingale Posterior Sampling for von Mises--Fisher Models
}

\author{Yi Xu, Xinye Chen, Sheng Jiang
\footnote{Email for correspondence:
jiangsheng@cuhk.edu.cn}
\\
School of Data Science\\ 
The Chinese University of Hong Kong, Shenzhen}
\date{\today}

\newtheorem{assumption}{Assumption}
\newtheorem{lemma}{Lemma}
\newtheorem{theorem}{Theorem}

\begin{document}

\maketitle
\abstract{
We develop a finite-horizon calibration method for raw-score martingale
posteriors, with von Mises--Fisher models as the main worked example. Starting
from the maximum likelihood estimator, predictive paths are generated by
simulating future observations from the current fitted model and updating the
natural parameter by unpreconditioned score increments. 
The main methodological step is to separate predictive simulation from
covariance calibration. Raw-score increments have Fisher-information covariance,
whereas Bernstein--von Mises calibration requires inverse-information covariance.
We therefore apply a terminal linear correction based on a local information
estimate. 
For more efficient implementation, 
we also introduce a hybrid version that replaces the omitted tail of the infinite predictive continuation by a Gaussian
approximation with matching leading-order quadratic variation. 
We prove fixed-\(n\) convergence, a finite-horizon approximation bound for the Gaussian tail, and a Bernstein--von Mises limit for the hybrid post-corrected
sampler under local regularity and consistent terminal calibration. 
Simulations show that tail correction reduces truncation-induced underdispersion, 
and an OSCAR ocean-current example illustrates local directional uncertainty summaries.
} 


\newpage 

\section{Introduction} 
\label{sec:introduction} 
Rotational and directional observations arise when the inferential target is an
orientation, heading, tilt, phase, or angular relationship. In two dimensions,
rotations in \(SO(2)\) are naturally identified with points on the unit circle
\(S^1\). In three dimensions, full rotations lie on the non-Euclidean manifold
\(SO(3)\), while rotation axes, directions, and oriented unit vectors lie on the
unit sphere \(S^2\). Examples include robot headings, wind directions, protein
backbone dihedral angles, neuronal phase shifts, crystallographic textures,
fault-plane focal mechanisms, spacecraft attitude, and joint kinematics. In this
work, we model planar rotations as circular data on \(S^1\); for
three-dimensional rotations, we focus on the associated rotation axis represented
as a unit vector on \(S^2\).

A canonical parametric model for directional observations on the unit sphere is
the von Mises--Fisher (vMF) family. For \(x\in S^{p-1}\), the \(p-1\)-dimensional unit
sphere in $\mathbb{R}^p$, the density with respect to surface measure is
\begin{equation}
f(x;\mu,\kappa)
=
c_p(\kappa)\exp\{\kappa \mu^\top x\},
\qquad
\mu \in S^{p-1},
\quad
\kappa \ge 0,
\nonumber 
\end{equation}
where
$ 
c_p(\kappa)
=
\frac{\kappa^{p/2-1}}
{(2\pi)^{p/2} I_{p/2-1}(\kappa)}
$
and \(I_\nu(\cdot)\) denotes the modified Bessel function of the first kind.
The parameter \(\mu\) specifies the mean direction, while \(\kappa\) controls
concentration about \(\mu\). The case \(\kappa=0\) corresponds to the uniform
distribution on the sphere, and larger values of \(\kappa\) correspond to
greater concentration around \(\mu\). The circular von Mises model and the
spherical von Mises--Fisher model are standard exponential-family baselines for
directional data inference.

Bayesian inference for vMF models has been developed along several lines, most of which rely on Markov chain Monte Carlo (MCMC) for posterior computation.
Early work focused primarily on the mean direction under a
known concentration parameter \(\kappa\) \citep{mardia1976bayesian}. 
Later work treated both direction and concentration, using canonical parameterizations
\citep{guttorp1988finding}, Gibbs sampling for circular data
\citep{damien1999full}, and sampling-importance-resampling for higher-dimensional
spherical data \citep{nunez2005bayesian}. 
Extensions to multimodal directional data include infinite vMF mixtures based on Dirichlet processes \citep{bangert2010using,straub2015dirichlet} and finite von Mises mixtures with an unknown number of components using trans-dimensional MCMC \citep{mulder2020bayesian}. 
\citet{forbes2015fast} address the computational bottleneck of MCMC by developing an efficient acceptance-rejection algorithm for sampling the concentration parameter in the circular von Mises model.


Martingale posterior distributions offer a prior-free alternative to uncertainty quantification, in the specific sense that no explicit prior distribution on the model parameter is specified: posterior draws are generated by simulating future predictive sequences and recording the limiting values of the inferential state, which evolves as a martingale.
This predictive viewpoint is closely related to earlier work on predictive constructions of Bayesian procedures and exchangeability \citep{fortini2012predictive,
fortini2025exchangeability,berti2025probabilistic}. The martingale posterior
formulation of \citet{fong2023martingale} has since been developed in several
directions, including asymptotic theory of parametric martingale posteriors \citep{fong2026asym},
quantile estimation and regression \citep{fong2025bayesian}, time-series models
\citep{moya2025martingale}, finite mixture models
\citep{rodriguez2025martingale}, 
and score-function-driven martingale posterior samplers 
\citep{cui2025martingale}.

For a regular parametric model \(p(x\mid\theta)\), let
$
s(x,\theta)=\nabla_\theta \log p(x\mid\theta)
$
denote the score function and let \(I(\theta)\) denote the Fisher information.
A calibrated parametric martingale posterior can be constructed by simulating
\(\widetilde X_{n+m}\sim p(\cdot\mid \theta_{n,m-1})\) and updating
\[
\theta_{n,m}
=
\theta_{n,m-1}
+
\gamma_{n,m}
I(\theta_{n,m-1})^{-1}
s(\widetilde X_{n+m},\theta_{n,m-1}),
\quad 
m=1,2,\ldots .
\]
The score has conditional mean zero, so the recursion is a martingale under the
predictive law. The Fisher preconditioner places the increments on the
inverse-information scale, which is the covariance scale appearing in
Bernstein--von Mises theory. This preconditioned recursion is therefore a useful
reference construction, but it requires Fisher-information evaluation and matrix
inversion along the predictive path.

\citet{cui2025martingale} observe that the mean-zero property of score functions is the minimal requirement for constructing martingales, and they study the basic convergence properties of score-driven predictive martingales. 
In this paper, we employ the same martingale identity but focus on two additional issues that arise in raw-score martingale posterior sampling (MPS): 
covariance scaling and finite-horizon truncation. 
In our application to vMF models, we adopt the natural parameter
$ 
c=\kappa\mu\in\mathbb R^p
$
and consider the raw-score recursion
\[
c_{n,m}
=
c_{n,m-1}
+
\gamma_{n,m}
s(\widetilde X_{n+m},c_{n,m-1}).
\]
This recursion remains a martingale because 
\(\mathbb{E}_{c}\{s(X,c)\}=0\). 
Its covariance, however, is not automatically on the Bernstein--von Mises scale: 
the conditional increment covariance is \(I(c_{n,m-1})\), whereas the target posterior covariance involves \(I(c^*)^{-1}\). 
Thus, for a raw-score martingale, martingale validity and posterior covariance calibration are separate questions.

First, forward simulation from vMF distributions can be carried out using existing methods such as those described by \citet{hoff2009simulation}. 
Hence the predictive path can be generated using
model simulation and score evaluation, while the inverse-information correction
is deferred to the terminal stage. 
In the basic vMF model the Fisher information is explicit, so this separation is not needed for feasibility; it is
nevertheless useful for understanding score-based martingale posteriors and for
settings where pathwise Fisher preconditioning is less convenient.

A second issue is finite simulation. 
The martingale posterior is defined through an infinite predictive continuation, but any implementation must truncate the path. 
If the raw-score recursion is stopped after \(M\) steps, the remaining
terms with tail weight
$
r_{n,M}
=
\sum_{m=M+1}^{\infty}(n+m)^{-2}
$ 
are not included. 
This can lead to variance loss when \(M\) is small relative to
\(n\). 
We therefore consider a hybrid implementation that simulates the first
\(M\) raw-score steps explicitly and replaces the ``missing" tail by a Gaussian
approximation, matching leading-order covariance. 
The resulting draw is then post-corrected using a terminal estimate of the local Fisher information.

The theoretical analysis is organized around these two calibration issues. First,
for fixed \(n\), the raw-score recursion is shown to converge almost surely under
standard martingale square-summability conditions. Second, for a finite
simulation horizon, the error from replacing the unresolved martingale tail by a
Gaussian approximation is bounded for smooth test functions under local
regularity conditions. Third, in the large-sample regime, the hybrid
post-corrected sampler is shown to satisfy a Bernstein--von Mises theorem when
the terminal calibration matrix consistently estimates the local Fisher
information. The consistency of this calibration matrix is an essential
condition. In particular, a single pathwise quadratic-variation estimate is not
generally consistent when \(M\) is fixed. Fixed explicit simulation depth is
therefore compatible with first-order Bernstein--von Mises calibration only when
the terminal information estimator is made consistent, for example by pooling
an independent calibration ensemble, or
using an analytic plug-in estimator.

The numerical experiments are used to illustrate these points. In circular von
Mises simulations, we compare the proposed martingale posterior samplers with a
Gibbs benchmark under the stated neutral limiting prior and examine coverage, interval length, and the effect of truncating the predictive path. 
The results illustrate the variance loss that can occur without a tail correction and the extent to which the hybrid approximation reduces this effect in the settings considered.

We also apply the method to satellite-derived OSCAR ocean surface current directions in a localized region of the California Current System. This real-data analysis should be viewed as an illustrative application rather than a complete model for the ocean-current field. In particular, the single-vMF working model is unimodal and therefore cannot capture important features of the data, such as multimodality, spatial dependence, or temporal variation. Nevertheless, it provides a useful first-order summary of local directional concentration and mean orientation, and allows a simple comparison across two seasonal snapshots. More realistic mixture models and spatial or spatio-temporal extensions of vMF models for this dataset are left for future work.

To facilitate reproducibility and reuse, we provide an open-source R
implementation of the proposed samplers in the package \texttt{bayesdir},
available at \url{https://github.com/sj156/bayesdir}. 
The remainder of the article is organized as follows.
Section~\ref{sec:methodology} reviews the score-based predictive resampling
framework and defines the martingale posterior samplers for the vMF model.
Section~\ref{sec:theory} gives the convergence result, the hybrid approximation
bound, and the Bernstein--von Mises theorem.
Section~\ref{sec:simulation} reports simulation experiments.
Section~\ref{sec:application} presents the ocean-current application.
Section~\ref{sec:conclusion} concludes with a discussion of extensions and
limitations. Technical proofs are collected in the Appendix.

\section{Methodology}  
\label{sec:methodology}

\subsection{Predictive Resampling and Martingale Posterior Distributions}
\label{sec:predictive_resampling}

The martingale posterior is a predictive, prior-free uncertainty distribution, where ``prior-free'' means that no explicit prior distribution on the model parameter is specified. Given the observed data \(X_{1:n}\), one specifies a sequence of one-step-ahead predictive distributions for future observations,
\[
    \widetilde X_{n+m}\mid \mathcal F_{n,m-1}
    \sim
    P_{n,m}(\cdot\mid \mathcal F_{n,m-1}),
    \qquad m\ge1,
\]
where
$
    \mathcal F_{n,m}
    =
    \sigma\{X_{1:n},\widetilde X_{n+1},\ldots,\widetilde X_{n+m}\}.
$
Let \(T_{n,m}\) be a data-dependent inferential state, such as a parameter estimate updated after the first \(m\) predictive draws. If \( \mathbb E(T_{n,m}\mid \mathcal F_{n,m-1})=T_{n,m-1} \) for \(m\ge1\), and \(T_{n,m}\to T_{n,\infty}\) almost surely, then the conditional law
\[
    \Pi_n^{\mathrm{MP}}(\cdot)
    =
    \mathcal L(T_{n,\infty}\in\cdot\mid X_{1:n})
\]
is called the martingale posterior. A posterior draw is obtained by simulating one complete future predictive path and recording the limiting inferential state \(T_{n,\infty}\). Thus uncertainty is represented through the variability of the unobserved future sequence rather than through a prior distribution on the parameter.

Predictive constructions of priors and posteriors have a rich lineage. \citet{fortini2012predictive,fortini2020quasi,fortini2025exchangeability} developed predictive constructions and their relationship with exchangeability and sequential learning, while \citet{berti2021class,berti2025probabilistic} subsequently formalized predictive models for Bayesian learning. The specific term \textit{martingale posterior} was introduced by \citet{fong2023martingale}, who rigorously defined posterior uncertainty as the limiting distribution generated by predictive resampling. This flexible framework has since been extended to quantile regression \citep{fong2025bayesian}, time-series modeling \citep{moya2025martingale}, finite mixture models \citep{rodriguez2025martingale}, and parametric martingale posteriors \citep{fong2026asym}. While conceptually related to the Bayesian bootstrap \citep{rubin1981bayesian}, the parametric bootstrap \citep{efron2012bayesian}, and generalized Bayesian updating \citep{bissiri2016general,knoblauch2022optimization}, the martingale posterior is distinct: 
uncertainty 
propagates through the predictive resampling which forms a martingale. 

For a parametric model $p(x \mid \theta)$, let $s(x, \theta) = \nabla_\theta \log p(x \mid \theta)$ denote the score function and $I(\theta)$ the Fisher information matrix. The calibrated parametric construction studied by \citet{fong2026asym} initializes at an estimator $\theta_{n,0}$, such as the maximum likelihood estimator (MLE), and iterates:
\begin{align}
    &\widetilde{X}_{n+m} \sim p(\cdot \mid \theta_{n,m-1}), \nonumber \\
    &\theta_{n,m} = \theta_{n,m-1} + \gamma_m I(\theta_{n,m-1})^{-1} s(\widetilde{X}_{n+m}, \theta_{n,m-1}), \quad m \ge 1, \label{eq:preconditioned_update}
\end{align}
where the deterministic step sizes satisfy $\sum_{m \ge 1} \gamma_m = \infty$ and $\sum_{m \ge 1} \gamma_m^2 < \infty$ \citep{lai2003stochastic}. 
Martingale posterior samples are obtained by independently repeating the predictive resampling procedure in~\eqref{eq:preconditioned_update}.

In \eqref{eq:preconditioned_update}, the Fisher information preconditioner maps score fluctuations onto the inverse-information scale. Under standard regularity conditions, this gives the covariance scale appearing in the classical Bernstein--von Mises limit. 
We refer to this construction as the Fisher-preconditioned (FP) martingale posterior sampler (MPS), or MPS-FP. 
A hybrid version of MPS-FP can also be constructed by replacing the unresolved preconditioned tail with a Gaussian approximation. However, both MPS-FP and its hybrid counterpart require explicit evaluation of the Fisher information, and MPS-FP further requires matrix inversion along the entire predictive path.

\subsection{Raw Score Martingale Posterior Samplers}
\label{sec:vmf_samplers}
We consider a raw-score construction in which Fisher preconditioning is not applied during predictive simulation. The recursion is
\begin{equation}
    \widetilde X_{n+m}\sim p(\cdot\mid\theta_{n,m-1}),
    \qquad
    \theta_{n,m}
    =
    \theta_{n,m-1}
    +
    \gamma_m\, s(\widetilde X_{n+m},\theta_{n,m-1}),
    \qquad m\ge1,
    \label{eq:raw_score_update_general}
\end{equation}
where the step sizes are $\gamma_m=(n+m)^{-1}$. Here and throughout, we suppress the dependence of $\gamma_m$ on the fixed initial sample size $n$. The choice satisfies $\sum_{m\ge1}\gamma_m=\infty$ and $w_{n,\infty}\equiv\sum_{m\ge1}\gamma_m^2=n^{-1}+O(n^{-2})$. Since the score has conditional mean zero under the model, the sequence $(\theta_{n,m})_{m\ge0}$ is a martingale. Whenever this recursion converges, its limit defines a martingale posterior draw $\theta_{n,\infty}$; Theorem~\ref{thm:as_convergence} gives sufficient conditions. 

The preconditioner in \eqref{eq:preconditioned_update} governs the covariance scale of the path. By the information identity $\mathrm{Var}_\theta\{s(X,\theta)\}=I(\theta)$, each increment in \eqref{eq:raw_score_update_general} has conditional covariance $\gamma_m^2 I(\theta_{n,m-1})$. The square-summable step sizes keep the path within $O_p(n^{-1/2})$ of its initialization, so summing the increments gives
\begin{equation}
    \operatorname{Var}\bigl(\theta_{n,\infty}-\theta_{n,0}\mid X_{1:n}\bigr)
    \approx
    w_{n,\infty}\, I(\theta_{n,0}),
    \label{eq:raw_scale}
\end{equation}
whereas the Bernstein--von Mises limit has covariance $n^{-1}I(\theta_{n,0})^{-1}$. The raw-score path therefore fluctuates on the information scale, and a separate calibration step is needed to reach the inverse-information scale.
\paragraph{Terminal covariance correction}
The calibration is a single affine map applied at the end of the path. For a positive definite estimator $\widehat I$ of $I(\theta_{n,0})$, define the returned draw
\begin{equation}
    \mathcal C(u;\widehat I)
    =
    \theta_{n,0}
    +
    \widehat I^{-1}(u-\theta_{n,0}),
    \label{eq:correction_map}
\end{equation}
where $u$ is the terminal draw from a raw-score path.\footnote{If $\widehat I$ is nearly singular, $\widehat I^{-1}$ in \eqref{eq:correction_map} may be replaced by the regularized inverse $(\widehat I+\varepsilon I_p)^{-1}$ for a small $\varepsilon>0$. This standard stabilization device affects neither the methodology nor the theory and is suppressed from the notation. 
}
When $u-\theta_{n,0}$ has covariance approximately $w_{n,\infty}I(\theta_{n,0})$ as in \eqref{eq:raw_scale} and $\widehat I\approx I(\theta_{n,0})$,
\[
    \operatorname{Var}\bigl\{\mathcal C(u;\widehat I)-\theta_{n,0}\bigr\}
    \approx
    \widehat I^{-1}\bigl\{w_{n,\infty}I(\theta_{n,0})\bigr\}\widehat I^{-1}
    \approx
    \bigl\{n^{-1}+O(n^{-2})\bigr\}\, I(\theta_{n,0})^{-1},
\]
which is the first-order Bernstein--von Mises covariance. Conditional on a fixed calibration matrix, the map is affine and therefore preserves non-Gaussian features generated by the early predictive updates. With path-specific calibration, randomness in the calibration matrix can additionally affect the shape of the returned distribution.

We use the term \emph{post-corrected martingale posterior sampler} for this calibrated transformation of draws generated by raw-score martingale paths. When \(\widehat I\) is estimated from a completed path or from the same pooled ensemble used to form the output sample, the corrected draws need not themselves be terminal values of an adapted martingale.
The estimator $\widehat I$ can be built from the simulated paths. Suppose $B$ independent predictive paths (chains) are generated, and write $s_{n,m}^{(b)}$ for the $m$-th score increment of chain $b$. Define the realized quadratic variation
\begin{equation}
    Q_{n,M}^{(b)}
    =
    \sum\nolimits_{m=1}^M
    \gamma_m^2
    s_{n,m}^{(b)}
    \{s_{n,m}^{(b)}\}^{\top},
    \qquad
    w_{n,M}
    =
    \sum\nolimits_{m=1}^M\gamma_m^2 .
    \label{eq:quadratic_variation_method}
\end{equation}
Its conditional expectation equals $\sum_{m=1}^M\gamma_m^2 I(\theta_{n,m-1}^{(b)})\approx w_{n,M}\,I(\theta_{n,0})$, which motivates the chain-specific calibration estimator
\begin{equation}
    \widehat I_{n,M}^{(b),\mathrm{path}}
    =
    Q_{n,M}^{(b)}/w_{n,M}.
    \label{eq:pathwise_information_estimator}
\end{equation}
This pathwise estimator requires no communication across chains and preserves the fully parallel structure of the sampler. Alternatively, one may use the pooled estimator
\begin{equation}
    \widehat I_{n,M,B}^{\mathrm{pool}}
    =
    \frac1B
    \sum\nolimits_{b=1}^B
    \widehat I_{n,M}^{(b),\mathrm{path}},
    \label{eq:pooled_information_estimator_methodology}
\end{equation}
which introduces a single synchronization step after all paths have been simulated but reduces Monte Carlo noise in the estimated local Fisher geometry. It is especially useful when \(M\) is small or when individual \(Q_{n,M}^{(b)}\) matrices are unstable. Because the returned draws then share a random calibration matrix, they are exchangeable but not strictly independent when the same paths are used both for calibration and output. The Bernstein--von Mises theorem in Section~\ref{sec:bvm} directly covers a consistent data-measurable estimator or a pooled estimator obtained from an independent calibration ensemble; for fixed $M$, the latter attains consistency as its ensemble size grows. Analytic information or an independent calibration ensemble also avoids shared-calibration dependence among the returned draws.

\paragraph{Recommended Usage.} In implementation the recursion is stopped after $M$ steps, which discards the tail weight
$
r_{n,M}=\sum_{m=M+1}^{\infty}\gamma_m^2=w_{n,\infty}-w_{n,M}
$
from the total predictive variance in \eqref{eq:raw_scale}. We consider two samplers built on the same raw-score path, summarized in Table~\ref{tab:vmf_samplers}. The \textit{raw-score post-corrected sampler} (MPS-R) applies \eqref{eq:correction_map} to the terminal state $\theta_{n,M}^{(b)}$. Since $w_{n,M}\approx M/\{n(n+M)\}$, the omitted tail produces variance attenuation unless $M\gg n$, so direct use of MPS-R requires a large truncation horizon relative to \(n\). The \textit{hybrid raw-score post-corrected sampler} (Hybrid MPS-R) appends a Gaussian tail $Z_{n,M}^{(b)}\sim N_p(0,r_{n,M}\widehat I_{\mathrm{cal}}^{(b)})$ before the correction, which restores the omitted variance weight at any fixed horizon $M$. Section~\ref{sec:vmf_specifics} specializes both samplers to the vMF model and gives the full algorithm.
\begin{table}[H]
\centering
\small
\captionsetup{width=0.90\linewidth, font=small}
\caption{Comparison of the raw score martingale posterior samplers. For both samplers, \(\widehat I_{\mathrm{cal}}^{(b)}\) may be either the pathwise estimator \eqref{eq:pathwise_information_estimator} or the pooled estimator \eqref{eq:pooled_information_estimator_methodology}.}
\label{tab:vmf_samplers}
\begin{tabular}{lll}
\toprule
Sampler
& Tail after \(M\)
& Returned draw \\
\midrule
MPS-R
&
None
&
\(\mathcal C(\theta_{n,M}^{(b)};
\widehat I_{\mathrm{cal}}^{(b)})\)
\\[0.5em]
Hybrid MPS-R
&
\(Z_{n,M}^{(b)}\sim \mathcal N_p(0,r_{n,M}\widehat I_{\mathrm{cal}}^{(b)})\)
&
\(\mathcal C(\theta_{n,M}^{(b)}+Z_{n,M}^{(b)};\widehat I_{\mathrm{cal}}^{(b)})\)
\\
\bottomrule
\end{tabular}
\end{table}

\subsection{Raw-Score Post-Corrected MPS for vMF Models}
\label{sec:vmf_specifics}

We now specialize the raw-score post-corrected martingale posterior samplers to the von Mises--Fisher setting. Both variants are summarized in Algorithm~\ref{alg:vmf_mps} at the end of the section. 

We adopt the natural parameterization 
$
c \equiv \kappa\mu \in \mathbb{R}^p,  
$ 
where $\kappa \equiv  \|c\|$ and 
$\mu \equiv \frac{c}{\|c\|}$, 
for mathematical convenience,
because the predictive updates are performed in the unconstrained Euclidean space $\mathbb{R}^p$, 
while the observations remain correctly constrained to the sphere $\mathbb{S}^{p-1}$.  
Under the natural parameterization, the score function is 
\begin{equation}
    s(x, c) = x - A_p(\|c\|) \frac{c}{\|c\|}, \quad A_p(\kappa) = \frac{I_{p/2}(\kappa)}{I_{p/2-1}(\kappa)}, \label{eq:vmf_score}
\end{equation}
where $I_\nu$ is the modified Bessel function of the first kind. We adopt the convention that $A_p(\|c\|)c/\|c\| = 0$ when $c=0$. Equivalently, $s(x,c)=x-\mathbb{E}_c(X)$; the score is the observed direction minus the model-implied mean direction. This closed form allows efficient evaluations during the predictive resampling path.

The Fisher information in the natural parameter is $I(c)=\mathrm{Var}_c(X)$. For $c\neq0$, writing $\kappa=\|c\|$ and $u=c/\kappa$, it has the form
\begin{equation}
    I(c)
    =
    \frac{A_p(\kappa)}{\kappa}(I_p-uu^\top)
    +
    A_p'(\kappa)uu^\top,
    \qquad
    A_p'(\kappa)
    =
    1-A_p(\kappa)^2-\frac{p-1}{\kappa}A_p(\kappa).
    \label{eq:vmf_fisher}
\end{equation}
At $c=0$, $I(0)=p^{-1}I_p$. 
For the vMF model, \(I(c)\) is explicit and low dimensional, so MPS-FP and its hybrid variant can be implemented directly and are included as references. 
We use raw-score post-corrected variants because they do not require Fisher preconditioning at each simulated step. This distinction is relevant for directional models in which the Fisher information is not available in closed form or is more burdensome to compute.

\subsubsection*{MPS-R: Raw-score post-corrected sampler}
The raw-score post-corrected martingale posterior sampler, abbreviated MPS-R, simulates the unpreconditioned score recursion
\begin{equation}
    \widetilde X_{n+m}^{(b)}
    \sim
    \operatorname{vMF}(c_{n,m-1}^{(b)}),
    \qquad
    c_{n,m}^{(b)}
    =
    c_{n,m-1}^{(b)}
    +
    \gamma_m s_{n,m}^{(b)},
    \label{eq:rs_pc_update}
\end{equation}
where $s_{n,m}^{(b)}=s(\widetilde X_{n+m}^{(b)},c_{n,m-1}^{(b)})$. After $M$ steps, choose either
$\widehat I_{\mathrm{cal}}^{(b)}=\widehat I_{n,M}^{(b),\mathrm{path}}$ or
$\widehat I_{\mathrm{cal}}^{(b)}=\widehat I_{n,M,B}^{\mathrm{pool}}$. The returned draw is
\begin{equation}
    c_{\mathrm{out}}^{(b)}
    =
    \mathcal C
    (c_{n,M}^{(b)};\widehat I_{\mathrm{cal}}^{(b)}).
    \label{eq:rs_pc_output}
\end{equation}
Thus MPS-R separates the martingale simulation and the covariance calibration: the path is generated using raw-score increments, and the inverse-information geometry is imposed at the terminal stage.  

\subsubsection*{Hybrid MPS-R: hybrid raw-score post-corrected sampler}
The hybrid raw-score post-corrected martingale posterior sampler, abbreviated Hybrid MPS-R, uses the same raw-score path as MPS-R but adds a Gaussian proxy for the unsimulated tail. 
After $M$ raw-score steps, define
$
r_{n,M}=
\sum\nolimits_{m=M+1}^{\infty}\gamma_m^2 .
$
Given the calibration estimator $\widehat I_{\mathrm{cal}}^{(b)}$, draw an independent Gaussian tail
\begin{equation}
    Z_{n,M}^{(b)}
    \sim
    N_p(0,r_{n,M}\widehat I_{\mathrm{cal}}^{(b)}).
    \label{eq:hybrid_tail_methodology}
\end{equation}
The returned draw is
\begin{equation}
    c_{\mathrm{out}}^{(b)}
    =
    \mathcal C
    (c_{n,M}^{(b)}+Z_{n,M}^{(b)};
    \widehat I_{\mathrm{cal}}^{(b)}).
    \label{eq:hrs_pc_output}
\end{equation} 
The Gaussian tail is used to restore the deterministic variance weight omitted by truncation, 
thus Hybrid MPS-R can recover the first-order Bernstein--von Mises covariance scale without requiring $M\gg n$, provided that the terminal calibration estimator is accurate for the local Fisher information.

\providecommand{\AlgComment}[1]{\hfill{\scriptsize\itshape \# #1}}

\begin{algorithm}[H]
\caption{Post-Corrected Martingale Posterior Sampler for vMF Models}
\label{alg:vmf_mps}
\small
\begin{algorithmic}[1]
\Require Observations \(X_{1:n}\subset\mathbb S^{p-1}\), initial estimator \(c_{n,0}\), depth \(M\), chains \(B\), sampler type \(\mathsf{type}\in\{\mathrm{MPS\text{-}R},\mathrm{Hybrid \ MPS\text{-}R}\}\), and calibration choice \(\mathrm{cal}\in\{\mathrm{pathwise},\mathrm{pooled}\}\).
\Ensure Posterior draws \(\{c_{\mathrm{out}}^{(b)}\}_{b=1}^B\).

\State \textbf{Global setup} \AlgComment{deterministic, done once}
\State Compute \(\gamma_m=(n+m)^{-1}\), \(w_{n,M}=\sum_{m=1}^M\gamma_m^2\); if \(\mathsf{type}=\mathrm{Hybrid \ MPS\text{-}R}\), compute \(r_{n,M}=\sum_{m=M+1}^{\infty}\gamma_m^2\).

\State \textbf{Step 1: simulate raw score paths} \AlgComment{parallel over \(b=1,\ldots,B\)}
\For{\(b=1,\ldots,B\)}
    \State \(c_{n,0}^{(b)}\gets c_{n,0}\), \(Q_{n,0}^{(b)}\gets0\).
    \For{\(m=1,\ldots,M\)}
        \State Draw \(\widetilde X_{n+m}^{(b)}\sim\operatorname{vMF}(c_{n,m-1}^{(b)})\); compute \(s_{n,m}^{(b)}=s(\widetilde X_{n+m}^{(b)},c_{n,m-1}^{(b)})\).
        \State \(Q_{n,m}^{(b)}\gets Q_{n,m-1}^{(b)}+\gamma_m^2s_{n,m}^{(b)}\{s_{n,m}^{(b)}\}^{\top}\). \AlgComment{quadratic variation}
        \State \(c_{n,m}^{(b)}\gets c_{n,m-1}^{(b)}+\gamma_m s_{n,m}^{(b)}\). \AlgComment{raw score update}
    \EndFor
\EndFor

\State \textbf{Step 2: construct calibration estimator}
\If{\(\mathrm{cal}=\mathrm{pooled}\)}
    \State \(\widehat I_{\mathrm{pool}}\gets B^{-1}\sum_{b=1}^B Q_{n,M}^{(b)}/w_{n,M}\). \AlgComment{shared across chains}
\EndIf

\State \textbf{Step 3: return posterior draws}
\For{\(b=1,\ldots,B\)}
    \State \(\widehat I^{(b)}\gets Q_{n,M}^{(b)}/w_{n,M}\) if \(\mathrm{cal}=\mathrm{pathwise}\), otherwise \(\widehat I^{(b)}\gets\widehat I_{\mathrm{pool}}\).
    \If{\(\mathsf{type}=\mathrm{MPS\text{-}R}\)}
        \State \(c_{\mathrm{out}}^{(b)}\gets\mathcal C(c_{n,M}^{(b)};\widehat I^{(b)})\). \AlgComment{post-correct raw path}
    \ElsIf{\(\mathsf{type}=\mathrm{Hybrid \ MPS\text{-}R}\)}
        \State Draw \(Z_{n,M}^{(b)}\sim N_p(0,r_{n,M}\widehat I^{(b)})\). \AlgComment{Gaussian tail}
        \State \(c_{\mathrm{out}}^{(b)}\gets\mathcal C(c_{n,M}^{(b)}+Z_{n,M}^{(b)};\widehat I^{(b)})\). \AlgComment{hybrid correction}
    \EndIf
\EndFor
\end{algorithmic}
\end{algorithm}

The algorithms described above are implemented in R in the open-source package
\texttt{bayesdir}, available at
\url{https://github.com/sj156/bayesdir}. The package includes routines for
simulating vMF predictive paths, computing raw-score updates, applying terminal
post-correction, and generating MPS-R and Hybrid MPS-R posterior draws.

\section{Theoretical Guarantees for Raw-Score Martingale Calibration} 
\label{sec:theory} 

This section addresses three issues. First, for fixed observed sample size \(n\), the raw-score recursion has an almost sure infinite-depth limit. This establishes that the raw predictive construction defines a valid limiting martingale posterior draw. Second, for a finite simulation horizon \(M\), the unresolved martingale tail can be approximated by a Gaussian random vector with matching leading-order quadratic variation. We give a smooth-test-function error bound for this hybrid approximation. Third, in the large-sample regime \(n\to\infty\), the Hybrid MPS-R satisfies a Bernstein--von Mises limit after centering at the initial MLE \(\theta_{n,0}\).

Although the proposed methods are developed for von Mises--Fisher inference, the main probabilistic arguments are not specific to directional data. We therefore state the results for a regular parametric model \(\{p(x\mid \theta):\theta\in\Theta\subseteq\mathbb R^p\}\). For the vMF model, the natural parameter \(c=\kappa\mu\in\mathbb R^p\) plays the role of \(\theta\), and the required regularity conditions hold locally on compact subsets of the natural-parameter space.

For completeness and ease of reading, we begin by briefly recalling the core mathematical notation used in our analysis. Although some of these definitions have been introduced earlier, consolidating them here ensures this section remains relatively self-contained and provides a convenient reference for the subsequent theoretical results.

Let \(X_{1:n}\) denote the observed data, and let \(\theta_{n,0}\) be the initial estimator, taken throughout to be the MLE. 
Conditional on \(X_{1:n}\), the raw-score predictive recursion generates, for \(m\ge 1\),
\begin{equation*}
    \widetilde X_{n+m}\mid \mathcal F_{n,m-1}
    \sim p(\cdot\mid \theta_{n,m-1}),
    \qquad
    \theta_{n,m}
    =
    \theta_{n,m-1}
    +
    \gamma_{m}s_{n,m}.
\end{equation*}
where \(\gamma_{m}=(n+m)^{-1}\), \(s_{n,m}=\nabla_\theta\log p(\widetilde X_{n+m}\mid \theta_{n,m-1})\), and \(\mathcal F_{n,m}=\sigma\{X_{1:n},\widetilde X_{n+1},\ldots,\widetilde X_{n+m}\}\).
For a finite horizon \(M\ge 1\), define
$
    D_{n,M}
    =
    \sum\nolimits_{m=1}^M \gamma_{m}s_{n,m},
$
$
    Q_{n,M}
    =
    \sum\nolimits_{m=1}^M \gamma_{m}^2s_{n,m}s_{n,m}^\top,
$
and 
$
    w_{n,M}
    =
    \sum\nolimits_{m=1}^M \gamma_{m}^2
$. 
Thus \(\theta_{n,M}=\theta_{n,0}+D_{n,M}\). 
The infinite-horizon deterministic weight is
$
w_{n,\infty}
\equiv 
\sum\nolimits_{m\ge 1}\gamma_{m}^2
=
\sum\nolimits_{m\ge 1}(n+m)^{-2}
=
n^{-1}+\mathcal O(n^{-2}).
$
The unresolved tail weight after \(M\) steps is \(r_{n,M}\equiv w_{n,\infty}-w_{n,M}=\sum\nolimits_{m>M}\gamma_{m}^2\). We also write \(A_{n,M}\equiv \sum\nolimits_{m>M}\gamma_{m}^3\), which controls the third-moment term in the Gaussian approximation of the martingale tail.

The pathwise quadratic-variation estimator of the local Fisher scale is \(\widehat I^{\mathrm{path}}_{n,M}=Q_{n,M}/w_{n,M}\). For large-sample calibration, however, we allow a generic positive definite estimator \(\widehat I_n\) of the local Fisher information. One important example is the pooled estimator based on \(B\) independent predictive paths 
\begin{equation}
\label{eq:pooled_information_estimator_theory}
    \widehat I_{n,M,B}^{\mathrm{pool}}
    =
    \frac{1}{B}
    \sum\nolimits_{b=1}^B
    \frac{Q_{n,M}^{(b)}}{w_{n,M}}.
\end{equation}
This distinction matters because a single pathwise estimator \(Q_{n,M}/w_{n,M}\) need not consistently estimate \(I(\theta_{n,0})\) when \(M\) is fixed, whereas the pooled estimator can be made consistent by increasing \(B\).

The nugget used in Algorithm \ref{alg:vmf_mps} is a numerical stabilization device and is asymptotically negligible when the nugget tends to zero or when the estimated information matrix is well conditioned. 
Given a positive definite information estimate \(\widehat I\), the post-corrected raw-score draw is 
$    \theta_{n,M}^{\mathrm{pc}}(\widehat I)
    =
    \theta_{n,0}
    +
    \widehat I^{-1}D_{n,M}.
$

\subsection{Convergence of the Raw-Score Recursion}
\label{sec:fixed_n_convergence}

We first fix \(n\) and study the limit as the predictive simulation depth \(M\to\infty\). In this regime, the observed data are held fixed, and all remaining randomness comes from the predictive resampling mechanism. Since \(\sum\nolimits_{m\ge1}\gamma_{m}^2<\infty\), the raw-score recursion is a square-summable martingale recursion. The following assumption gives a convenient fixed-\(n\) convergence condition.

\begin{assumption}[Fixed-\(n\) martingale conditions]
\label{assum:martingale_convergence}
For fixed \(n\), conditional on \(X_{1:n}\), assume:
\begin{enumerate}[(i)]
    \item The predictive scores are martingale differences: \(\mathbb E(s_{n,m}\mid \mathcal F_{n,m-1})=0\) for all \(m\ge 1\).
    \item The accumulated second moment is conditionally integrable:
    \[
        \mathbb E\!\left[
            \sum\nolimits_{m\ge1}\gamma_m^2\|s_{n,m}\|^2
            \,\middle|\, X_{1:n}
        \right]<\infty
        \quad\text{a.s.}
    \]
    \item The realized quadratic variation satisfies \(Q_{n,M}\to Q_{n,\infty}\) a.s., where \(Q_{n,\infty}\) is finite and positive definite.
\end{enumerate}
\end{assumption}

\begin{theorem}[Fixed-\(n\) convergence of the raw-score martingale]
\label{thm:as_convergence}
Suppose Assumption \ref{assum:martingale_convergence} holds. Then, for fixed \(n\), \(D_{n,M}=\sum\nolimits_{m=1}^M\gamma_{m}s_{n,m}\) converges almost surely and in \(L^2\) under the conditional predictive law to a finite random variable \(D_{n,\infty}\). Consequently,
\[
    \theta_{n,M}
    \equiv 
    \theta_{n,0}+D_{n,M}
    \to
    \theta_{n,\infty}
    \equiv 
    \theta_{n,0}+D_{n,\infty}
    \qquad \text{a.s.}
\]
Moreover, \(\widehat I^{\mathrm{path}}_{n,M}=Q_{n,M}/w_{n,M} 
\to \widehat I^{\mathrm{path}}_{n,\infty}\equiv Q_{n,\infty}/w_{n,\infty}\) a.s. 

If \(Q_{n,\infty}\) is positive definite, 
then the pathwise post-corrected draw \(\theta_{n,M}^{\mathrm{pc}}=
\theta_{n,0}+(\widehat I^{\mathrm{path}}_{n,M})^{-1}D_{n,M}\) converges almost surely to
$
\theta_{n,\infty}^{\mathrm{pc}}
=
\theta_{n,0}
+
\left(\widehat I^{\mathrm{path}}_{n,\infty}\right)^{-1}D_{n,\infty}.
$
\end{theorem}

\begin{proof}[proof sketch]
The formal proof is deferred to Appendix \ref{app:proof_thm1}. 
    Let \(Y_{n,m}=\gamma_{m}s_{n,m}\). By Assumption 1(i),
\((D_{n,M},\mathcal F_{n,M})\), with \(D_{n,M}=\sum_{m=1}^M Y_{n,m}\), is a
martingale. Assumption 1(ii) gives conditional \(L^2\)-boundedness, so the
martingale convergence theorem yields \(D_{n,M}\to D_{n,\infty}\) almost surely
and in conditional \(L^2\). Hence \(\theta_{n,M}=\theta_{n,0}+D_{n,M}\to
\theta_{n,\infty}:=\theta_{n,0}+D_{n,\infty}\) a.s.

By Assumption 1(iii), \(Q_{n,M}\to Q_{n,\infty}\) a.s., while
\(w_{n,M}\to w_{n,\infty}\in(0,\infty)\). Therefore
\[
\widehat I^{\mathrm{path}}_{n,M}
\equiv 
\frac{Q_{n,M}}{w_{n,M}}
\to
\widehat I^{\mathrm{path}}_{n,\infty}
\equiv 
\frac{Q_{n,\infty}}{w_{n,\infty}}
\qquad \text{a.s.}
\]
If \(Q_{n,\infty}\succ0\), then
\(\widehat I^{\mathrm{path}}_{n,\infty}\succ0\), and continuity of matrix
inversion gives
\[
\theta^{\mathrm{pc}}_{n,M}
=
\theta_{n,0}
+
\left(\widehat I^{\mathrm{path}}_{n,M}\right)^{-1}D_{n,M}
\to
\theta_{n,0}
+
\left(\widehat I^{\mathrm{path}}_{n,\infty}\right)^{-1}D_{n,\infty}
=
\theta^{\mathrm{pc}}_{n,\infty}
\qquad \text{a.s.}
\]
\end{proof}

Theorem \ref{thm:as_convergence} is an existence result for the infinite-depth raw-score construction. It shows that the raw predictive recursion has a well-defined limiting distribution and 
that the pathwise post-correction is also well defined whenever the limiting quadratic variation is nonsingular. 

For fixed \(n\), the distribution of \(\theta_{n,\infty}\) is generally non-Gaussian. Early predictive updates may contribute a non-negligible part of the total path variation, especially when \(n\) is small. Thus the raw-score construction can retain finite-sample non-Gaussian features of the predictive distribution. 

\paragraph{Verification of Assumption 1 for the vMF model.}
For the vMF model, condition (i) trivially holds, and 
condition (ii) holds automatically because the natural-parameter score \(s(x,c)=x-A_p(\|c\|)c/\|c\|\) satisfies \(\|s(x,c)\|\le 1+A_p(\|c\|)\le 2\). 

For condition (iii), 
$Q_{n,M}$ 
converges almost surely because its operator norm 
\[
\sum\nolimits_{m=1}^\infty
\left\|\gamma_{m}^2s_{n,m}s_{n,m}^\top\right\|_{\mathrm{op}}
=
\sum\nolimits_{m=1}^\infty
\gamma_{m}^2\left\|s_{n,m}\right\|^2
\le
4\sum\nolimits_{m=1}^\infty (n+m)^{-2}
<\infty.
\]
Thus \(Q_{n,M}\to Q_{n,\infty}\) a.s. for some finite nonnegative definite
matrix \(Q_{n,\infty}\).
The positive definite part is algebraic, hence left to the appendix. 

\subsection{Finite-Horizon Hybrid Tail Approximation}
\label{sec:hybrid_bounds}

The infinite-depth limit in Theorem \ref{thm:as_convergence} is idealized.
In computation, the predictive path must be stopped after finitely many steps.
If the MPS-R is stopped at horizon \(M\), the residual tail weight
\(r_{n,M}=\sum_{m>M}\gamma_{m}^2\) is omitted. This produces deterministic
variance shrinkage. The Hybrid MPS-R avoids this loss by replacing the
unresolved martingale tail with a Gaussian proxy.

Define the exact raw tail after time \(M\) by
$
    T_{n,M}=\sum_{m>M}\gamma_{m}s_{n,m}.
$ 
Conditional on \(\mathcal F_{n,M}\), this is a sum of future martingale
differences. Its mean is zero, and its leading conditional covariance is
obtained by freezing the Fisher information at the terminal value:
\[
    \mathbb E\{T_{n,M}T_{n,M}^{\top}\mid \mathcal F_{n,M}\}
    =
    \mathbb E\left[
        \sum\nolimits_{m>M}\gamma_{m}^2 I(\theta_{n,m-1}) 
        \,\middle|\, \mathcal F_{n,M}
    \right]
    \approx
    r_{n,M}I(\theta_{n,M}).
\]
Thus the natural finite-horizon approximation is a Gaussian tail with
covariance \(r_{n,M}I(\theta_{n,M})\), or with \(I(\theta_{n,M})\) replaced by a
terminal estimator \(\widehat I_{n,M}\).

Given \(\widehat I_{n,M}\), the ideal tail-corrected draw and its hybrid
approximation are
\[
\theta_{n,M}^{\mathrm{tail,pc}}
=
\theta_{n,0}
+
\widehat I_{n,M}^{-1}
\{D_{n,M}+T_{n,M}\},
\quad
\theta_{n,M}^{\mathrm{hyb,pc}}
=
\theta_{n,0}
+
\widehat I_{n,M}^{-1}
\{D_{n,M}+Z_{n,M}\},
\]
where
$
    Z_{n,M}\mid \mathcal F_{n,M}
    \sim
    \mathcal N_p(0,r_{n,M}\widehat I_{n,M}).
$
After post-correction, the Gaussian tail contributes leading covariance
\[
    \widehat I_{n,M}^{-1}
    (r_{n,M}\widehat I_{n,M})
    \widehat I_{n,M}^{-1}
    =
    r_{n,M}\widehat I_{n,M}^{-1},
\]
which is precisely the inverse-information-scale variance omitted by
finite-horizon truncation. The following theorem quantifies the error from
replacing the true martingale tail \(T_{n,M}\) by the Gaussian proxy
\(Z_{n,M}\). The comparison is made through smooth bounded test functions, so
the result is a weak-distribution approximation. 

\begin{assumption}[Local regularity]
\label{assum:hybrid_regularity}
Let \(K\subset\Theta\) be compact. On \(K\), assume:
\begin{enumerate}[(i)]
    \item \(\mathbb E_\theta\{s(X,\theta)\}=0\).
    \item \(\sup_{x,\theta\in K}\|s(x,\theta)\|\le S_K<\infty\).
    \item \(\|I(\theta)-I(\theta')\|_{\mathrm{op}}
    \le L_K\|\theta-\theta'\|\) for all \(\theta,\theta'\in K\).
    \item \(\lambda_{\min}\{I(\theta)\}\ge \lambda_K>0\) for all
    \(\theta\in K\).
\end{enumerate}
\end{assumption}

The compact set \(K\) is a localization device. The approximation below assumes,
under the conditional predictive law at time \(M\), that the continuation stays
in a region where the score is bounded, the Fisher information is Lipschitz, and
the information matrix is uniformly nonsingular. This is an assumption on the
conditional law, rather than conditioning on an event determined by future
draws. An unrestricted process can be handled by a stopped localization and an
explicit exit-probability term, as described after the theorem.

\begin{theorem}[Finite-horizon hybrid approximation]
\label{thm:hybrid_bound}
Suppose Assumption \ref{assum:hybrid_regularity} holds on a compact set \(K\).
Assume that, conditionally on \(\mathcal F_{n,M}\), the predictive continuation
after time \(M\) remains in \(K\) almost surely. Let \(\widehat I_{n,M}\) be a positive definite
information estimator, measurable with respect to the simulated path up to time
\(M\), satisfying
\[
    \|\widehat I_{n,M}-I(\theta_{n,M})\|_{\mathrm{op}}
    \le
    \varepsilon_{n,M},
    \qquad
    \lambda_{\min}(\widehat I_{n,M})\ge \lambda_K/2 .
\]
Then, for every \(\varphi\in C_b^3(\mathbb R^p)\),
\begin{equation}
\label{eq:hybrid_bound}
\begin{aligned}
    &
    \left|
    \mathbb E\{\varphi(\theta_{n,M}^{\mathrm{tail,pc}})\mid \mathcal F_{n,M}\}
    -
    \mathbb E\{\varphi(\theta_{n,M}^{\mathrm{hyb,pc}})\mid \mathcal F_{n,M}\}
    \right|
    \\
    &\qquad\le
    C_p
    \left[
    \frac{\|D^2\varphi\|_\infty}{\lambda_K^2}
    \left\{
        L_KS_K r_{n,M}^{3/2}
        +
        r_{n,M}\varepsilon_{n,M}
    \right\}
    +
    \frac{\|D^3\varphi\|_\infty}{\lambda_K^3}
    S_K^3 A_{n,M}
    \right],
\end{aligned}
\end{equation}
where \(C_p\) depends only on the dimension \(p\).
\end{theorem}

\begin{proof}[Proof sketch]
The formal proof is deferred to Appendix \ref{app:proof_thm2}. 
Let \(H_M=I(\theta_{n,M})\) and define
$
    f(y)
    =
    \varphi\!\left[
        \theta_{n,0}
        +
        \widehat I_{n,M}^{-1}\{D_{n,M}+y\}
    \right].
$
Since \(\lambda_{\min}(\widehat I_{n,M})\ge \lambda_K/2\),
\[
    \|D^2 f\|_\infty
    \le
    C\lambda_K^{-2}\|D^2\varphi\|_\infty,
    \qquad
    \|D^3 f\|_\infty
    \le
    C\lambda_K^{-3}\|D^3\varphi\|_\infty .
\]
Let \(G_M\sim \mathcal N_p(0,r_{n,M}H_M)\). A smooth martingale
Lindeberg replacement, conditionally on \(\mathcal F_{n,M}\), gives
\[
\left|
\mathbb E\{f(T_{n,M})\mid \mathcal F_{n,M}\}
-
\mathbb E\{f(G_M)\mid \mathcal F_{n,M}\}
\right|
\le
C_p\left\{
\|D^2f\|_\infty L_KS_K r_{n,M}^{3/2}
+
\|D^3f\|_\infty S_K^3A_{n,M}
\right\}.
\]
The first term is the cost of freezing
\(I(\theta_{n,m-1})\) at \(I(\theta_{n,M})\) along the tail; the second is the
third-order Gaussian-approximation remainder. 

Next, since
\(Z_{n,M}\sim \mathcal N_p(0,r_{n,M}\widehat I_{n,M})\), the Gaussian covariance
perturbation bound yields
\[
\left|
\mathbb E\{f(G_M)\mid \mathcal F_{n,M}\}
-
\mathbb E\{f(Z_{n,M})\mid \mathcal F_{n,M}\}
\right|
\le
C_p\|D^2f\|_\infty r_{n,M}
\|\widehat I_{n,M}-H_M\|_{\mathrm{op}} .
\]
Using
\(\|\widehat I_{n,M}-H_M\|_{\mathrm{op}}\le \varepsilon_{n,M}\) and substituting
the derivative bounds for \(f\) gives \eqref{eq:hybrid_bound}.
\end{proof}

Theorem \ref{thm:hybrid_bound} gives a finite-horizon justification for the
Hybrid MPS-R. Its bound separates three effects. The term
\(L_KS_Kr_{n,M}^{3/2}\) is the covariance-freezing error caused by the movement
of the predictive parameter along the unresolved tail. The term
\(S_K^3A_{n,M}\) is the martingale Gaussian-approximation error, analogous to a
Berry--Esseen third-moment contribution. The term
\(r_{n,M}\varepsilon_{n,M}\) is the price of using the estimated information
matrix \(\widehat I_{n,M}\) instead of \(I(\theta_{n,M})\).

For the step size \(\gamma_{m}=(n+m)^{-1}\),
$r_{n,M} =
\sum\nolimits_{m>M}(n+m)^{-2}
\asymp (n+M)^{-1},
$
and 
$
A_{n,M}
    =
    \sum\nolimits_{m>M}(n+m)^{-3}
    \asymp
    (n+M)^{-2}.
$
Thus, for fixed \(n\), the Gaussian tail approximation becomes more accurate
as the explicitly simulated portion of the predictive path grows.

The bound is conditional on \(\mathcal F_{n,M}\), which is the natural
formulation for the algorithm: after the first \(M\) steps have been simulated,
\(D_{n,M}\), \(\theta_{n,M}\), and \(\widehat I_{n,M}\) are fixed, and the only
remaining randomness is the unresolved future tail. If the right-hand side of
\eqref{eq:hybrid_bound} is deterministic, this gives a uniform bound on the
conditional approximation error over all localized histories.  
For the unrestricted process, let
\(\tau_K=\inf\{m\ge M:\theta_{n,m}\notin K\}\). Applying the comparison to a
localized continuation that agrees with the original process up to \(\tau_K\),
and then comparing the localized and unrestricted laws, adds an exit-probability
term. If \(\mathcal E_{K,M}=\{\tau_K=\infty\}\), the resulting unconditional
bound has the form
\[
\left|
\mathbb E\{\varphi(\theta_{n,M}^{\mathrm{tail,pc}})\}
-
\mathbb E\{\varphi(\theta_{n,M}^{\mathrm{hyb,pc}})\}
\right| 
\le
B_{n,M}
+
2\|\varphi\|_\infty\,\mathbb P(\mathcal E_{K,M}^c),
\]
where \(B_{n,M}\) denotes the right-hand side of
\eqref{eq:hybrid_bound}. Hence localization is asymptotically harmless whenever
the probability of leaving the regularity region is negligible.

The estimation error of information matrix, \(\varepsilon_{n,M}\), depends on how
\(\widehat I_{n,M}\) is constructed. For the pooled estimator
\eqref{eq:pooled_information_estimator_theory}, define
$
    N_{\mathrm{eff}}(n,M)
    =
    \frac{w_{n,M}^2}{\sum_{m=1}^M\gamma_{m}^4}.
$
Under bounded scores and local Lipschitz continuity of \(I(\cdot)\), standard
matrix concentration gives the representative rate
\[
    \left\|
    \widehat I_{n,M,B}^{\mathrm{pool}}
    -
    I(\theta_{n,0})
    \right\|_{\mathrm{op}}
    =
    \mathcal O_p
    \left[
        \sqrt{\frac{\log p}{B\,N_{\mathrm{eff}}(n,M)}}
        +
        \sqrt{w_{n,M}}
    \right],
\]
up to constants depending on \(S_K\) and \(L_K\). Since
\[
    \|I(\theta_{n,M})-I(\theta_{n,0})\|_{\mathrm{op}}
    \le
    L_K\|\theta_{n,M}-\theta_{n,0}\|
    =
    \mathcal O_p(\sqrt{w_{n,M}}),
\]
the same rate controls the error relative to \(I(\theta_{n,M})\). If \(M\) is
fixed, then \(N_{\mathrm{eff}}(n,M)\asymp M\), so the pooled estimator becomes
consistent when \(B\to\infty\). Hence the explicit simulation depth \(M\) need
not scale with \(n\), but the terminal information estimator must still be
accurate.

\paragraph{Specialization to the von Mises--Fisher model.}
For the vMF model, localization has a simple interpretation. The score is
globally bounded by \(2\), but the Fisher information is not uniformly
nonsingular over all \(c\in\mathbb R^p\). A natural localized set is
\(K_R=\{c:\|c\|\le R\}\), which corresponds to bounding the concentration
parameter. 
On \(K_R\), the Fisher information is Lipschitz and uniformly nonsingular. Thus Theorem \ref{thm:hybrid_bound} applies to vMF predictive paths whose concentration remains in a bounded range.


\begin{corollary}[Local regularity for the vMF model]
\label{cor:vmf_bound}
For every \(R<\infty\), Assumption
\ref{assum:hybrid_regularity} holds on
\[
    K_R=\{c\in\mathbb R^p:\|c\|\le R\}.
\]
In particular, \(I(c)\) is Lipschitz continuous and uniformly nonsingular on
\(K_R\). 
The eigenvalues of \(I(c)\) are \(A_p(\kappa)/\kappa\) in directions
orthogonal to \(\mu\), 
and \(A_p'(\kappa)\) in the radial direction \(\mu\). 

With the continuous extensions
$
    \frac{A_p(\kappa)}{\kappa}\bigg|_{\kappa=0}
    =
    A_p'(0)
    =
    \frac1p,
$
we have \(I(0)=p^{-1}I_p\). 
Hence
$
    \lambda_R
    =
    \inf_{0\le \kappa\le R}
    \min
    \left\{
        \frac{A_p(\kappa)}{\kappa},
        A_p'(\kappa)
    \right\}
    >
    0.
$
Consequently, Theorem \ref{thm:hybrid_bound} applies to vMF predictive paths
localized in \(K_R\).
\end{corollary}

The proof is deferred to the appendix. 
For the vMF model, the localization event in Theorem
\ref{thm:hybrid_bound} has the simple interpretation
$
    \mathcal E_{R,M}
    =
    \left\{
        \sup_{m\ge M}\|c_{n,m}\|\le R
    \right\}.
$
It excludes predictive continuations along which the concentration parameter
becomes arbitrarily large. On \(\mathcal E_{R,M}\), the theorem applies with
\(S_K=2\), \(\lambda_K=\lambda_R\), and \(L_K=L_R\), the Lipschitz constant of
\(I(c)\) on \(K_R\).

This localization is mainly a proof device. If the true natural parameter
\(c_\ast\) is finite and \(c_{n,0}\to c_\ast\), then for any
\(R>\|c_\ast\|\) sufficiently large, the probability of leaving \(K_R\) is
small in large samples. Indeed, since the vMF score is bounded,
\[
    \mathbb E\left(
        \sup_{m\ge 0}\|c_{n,m}-c_{n,0}\|^2
        \,\middle|\, X_{1:n}
    \right)
    \lesssim
    \sum\nolimits_{m=1}^\infty (n+m)^{-2} 
    =
    O(n^{-1}),
\]
by Doob's inequality for the martingale displacement. Thus the predictive path
moves only \(O_p(n^{-1/2})\) from its initialization. Consequently, for bounded
concentration regimes, localization does not affect the first-order
large-sample calibration.

The corollary also clarifies why a global statement would be incorrect.
Although the score remains bounded for all \(c\), the Fisher information
degenerates as \(\|c\|\to\infty\). For example,
\(A_p(\kappa)/\kappa\to0\), and \(A_p'(\kappa)\to0\), as
\(\kappa\to\infty\). Hence no positive lower bound on the eigenvalues of
\(I(c)\) can hold uniformly over all of \(\mathbb R^p\). 

\subsection{Bernstein--von Mises Calibration}
\label{sec:bvm}

We now turn to the large-sample regime \(n\to\infty\). A Bernstein--von Mises theorem is a conditional statement about the posterior distribution given the observed data. Therefore the appropriate centering is the efficient data-dependent estimator \(\theta_{n,0}\), rather than the fixed true parameter \(\theta^*\). The target limit is
\[
    \mathcal L
    \left[
        \sqrt n\{\theta-\theta_{n,0}\}
        \mid X_{1:n}
    \right]
    \Rightarrow
    \mathcal N_p\{0,I(\theta^*)^{-1}\}.
\]
Equivalently, the conditional posterior covariance should be \(n^{-1}I(\theta^*)^{-1}\) to first order.

\begin{assumption}[Large-sample regularity]
\label{assum:mle_bvm}
Let \(\theta^*\) be an interior point of \(\Theta\). Assume:
\begin{enumerate}[(i)]
    \item The initial MLE is consistent and asymptotically normal: \(\theta_{n,0}\xrightarrow{p}\theta^*\) and \(\sqrt n(\theta_{n,0}-\theta^*)\xrightarrow{d}\mathcal N_p\{0,I(\theta^*)^{-1}\}\).
    \item \(I(\theta^*)\) is positive definite, and Assumption \ref{assum:hybrid_regularity} holds on a compact neighborhood \(K\) of \(\theta^*\).
    \item The information estimator used in the terminal correction is measurable with respect to the observed data or is obtained from an independent calibration ensemble. It satisfies \(\Delta_{I,n}\equiv \|\widehat I_n-I(\theta_{n,0})\|_{\mathrm{op}}=o_p(1)\), and \(\lambda_{\min}(\widehat I_n)\ge \lambda_{\min}\{I(\theta^*)\}/2\) with probability tending to one.
\end{enumerate}
\end{assumption}

Assumption \ref{assum:mle_bvm}(iii) is essential. It is satisfied, for example, by a pooled estimator \(\widehat I_{n,M,B}^{\mathrm{pool}}\) computed from an independent calibration ensemble when \(B\,N_{\mathrm{eff}}(n,M)\to\infty\) and the predictive paths remain in a compact neighborhood of \(\theta^*\) with probability tending to one. Thus fixed \(M\) is compatible with Bernstein--von Mises calibration, provided that enough independent predictive paths are used to estimate the local Fisher scale consistently. Reusing the output paths themselves for pooled calibration is computationally convenient, but the resulting shared-matrix dependence is not the conditioning regime asserted by the theorem.

\begin{theorem}[Bernstein--von Mises limit for the Hybrid MPS-R]
\label{thm:bvm}
Suppose Assumptions \ref{assum:hybrid_regularity} and \ref{assum:mle_bvm} hold. Let \(M=M_n\) be a deterministic simulation horizon satisfying \(M_n/n\to 0\), which includes the practically important case of fixed \(M\). Define
$
    \theta_n^{\mathrm{hyb,pc}}
    =
    \theta_{n,0}
    +
    \widehat I_n^{-1}
    \{D_{n,M_n}+Z_{n,M_n}\} 
$, and 
$    Z_{n,M_n}\mid \mathcal F_{n,M_n}
    \sim
    \mathcal N_p(0,r_{n,M_n}\widehat I_n).
$
Then
\begin{equation}
\label{eq:bvm_centered}
    \mathcal L
    \left[
        \sqrt n
        \{\theta_n^{\mathrm{hyb,pc}}-\theta_{n,0}\}
        \mid X_{1:n}
    \right]
    \Rightarrow
    \mathcal N_p\{0,I(\theta^*)^{-1}\}
\end{equation}
in probability. 

In particular, the following smooth test-function bound holds.
For every \(\psi\in C_b^3(\mathbb R^p)\),
\[
    \left|
    \mathbb E
    \left[
        \psi
        \left\{
            \sqrt n
            (\theta_n^{\mathrm{hyb,pc}}-\theta_{n,0})
        \right\}
        \mid X_{1:n}
    \right]
    -
    \mathbb E\{\psi(Z)\}
    \right|
    =
    \mathcal O_p
    \left(
        \Delta_{I,n}
        +
        \sqrt{\frac{M_n}{n}}
        +
        n^{-1/2}
    \right),
\]
where \(Z\sim \mathcal N_p\{0,I(\theta^*)^{-1}\}\).
\end{theorem}

\begin{proof}[Proof sketch]
The formal proof is deferred to Appendix \ref{app:proof_thm3}.  
Let \(M=M_n\), \(w_n=w_{n,M_n}\), \(r_n=r_{n,M_n}\), and
\(\widehat I=\widehat I_n\). Decompose
\[
    \sqrt n(\theta_n^{\mathrm{hyb,pc}}-\theta_{n,0})
    =
    R_n+G_n,
    \qquad
    R_n=\sqrt n\,\widehat I^{-1}D_{n,M},
    \quad
    G_n=\sqrt n\,\widehat I^{-1}Z_{n,M}.
\]
On the high-probability event where the predictive path remains in the local
regularity set and \(\widehat I\) is uniformly nonsingular, martingale
orthogonality gives
\(\mathbb E\{\|D_{n,M}\|^2\mid X_{1:n}\}\le Cw_n\). Hence
\(\mathbb E\{\|R_n\|\mid X_{1:n}\}\le C\sqrt{nw_n}\le C\sqrt{M_n/n}\), so the
explicitly simulated part is negligible on the \(\sqrt n\)-scale when
\(M_n/n\to0\).

Next, we identify the Gaussian component. 
Conditionally on the simulated path, 
\(G_n\sim \mathcal N_p(0,\Sigma_n)\), where
\(\Sigma_n=nr_n\widehat I^{-1}\). Since
\(w_{n,\infty}=\sum_{m\ge1}(n+m)^{-2}=n^{-1}+O(n^{-2})\) and
\(r_n=w_{n,\infty}-w_n\), we have
\(nr_n=1+O(n^{-1}+M_n/n)\). Moreover, by consistency of the calibration matrix,
local Lipschitz continuity of \(I(\cdot)\), and
\(\theta_{n,0}-\theta_\ast=O_p(n^{-1/2})\),
\[
    \|\widehat I^{-1}-I(\theta_\ast)^{-1}\|_{\mathrm{op}}
    =
    O_p(\Delta_{I,n}+n^{-1/2}).
\]
Therefore
$
    \|\Sigma_n-I(\theta_\ast)^{-1}\|_{\mathrm{op}}
    =
    O_p\!\left(\Delta_{I,n}+n^{-1/2}+M_n/n\right).
$

For any \(\psi\in C_b^2(\mathbb R^p)\), the contribution of \(R_n\) is bounded
by \(\|D\psi\|_\infty \mathbb E(\|R_n\|\mid X_{1:n})\), which is
\(O_p(\sqrt{M_n/n})\). The difference between \(G_n\) and
\(Z_\ast\sim\mathcal N_p\{0,I(\theta_\ast)^{-1}\}\) is controlled by a Gaussian
covariance perturbation bound. Combining the two bounds gives
\[
\left|
\mathbb E\!\left[
    \psi\{\sqrt n(\theta_n^{\mathrm{hyb,pc}}-\theta_{n,0})\}
    \,\middle|\, X_{1:n}
\right]
-
\mathbb E\{\psi(Z_\ast)\}
\right|
=
O_p\!\left(
    \Delta_{I,n}
    +
    \sqrt{M_n/n}
    +
    n^{-1/2}
\right).
\]
Since \(\Delta_{I,n}=o_p(1)\) and \(M_n/n\to0\), the conditional weak
convergence follows. If \(\widehat I_n\) is obtained from an independent
calibration ensemble, the same conclusion follows by first conditioning on the
ensemble and then applying the tower property.
\end{proof}

Theorem \ref{thm:bvm} establishes the first-order frequentist calibration of the Hybrid MPS-R. The statement is conditional on the observed data: after centering at \(\theta_{n,0}\) and scaling by \(\sqrt n\), the hybrid post-corrected posterior converges to the usual inverse-Fisher Gaussian limit. This is the Bernstein--von Mises target for a regular parametric posterior.

When \(M_n/n\to0\), the explicitly simulated raw-score component is
\(o_p(1)\) on the \(\sqrt n\)-scale, so the Gaussian tail determines the
first-order limit. The theorem should therefore be interpreted as establishing
first-order calibration, rather than an asymptotic improvement over the direct
Wald or Laplace approximation
\(N_p\{\theta_{n,0},n^{-1}I(\theta_{n,0})^{-1}\}\). The potential distinction is
finite-sample: early predictive updates can retain non-Gaussian features when
their contribution to total variation is non-negligible.

The theorem also explains the role of the simulation horizon. For the MPS-R stopped at \(M\), the accumulated variance weight is
$
w_{n,M}
=  
\sum\nolimits_{m=1}^M(n+m)^{-2}
\approx
\frac{M}{n(n+M)}.
$
If \(M\) is fixed while \(n\to\infty\), then \(w_{n,M}=\mathcal O(M/n^2)\), which is too small relative to the required posterior scale \(n^{-1}\). Thus a non-hybrid truncated path would require \(M\gg n\) to recover the full Bernstein--von Mises variance.

The Hybrid MPS-R removes this requirement by adding the missing deterministic tail weight:
\[
    w_{n,M}+r_{n,M}
    =
    w_{n,\infty}
    =
    n^{-1}+\mathcal O(n^{-2}).
\]
Consequently, the explicit simulation depth \(M\) does not need to grow proportionally to \(n\) in order to recover the correct first-order posterior variance. The terminal information estimator \(\widehat I_n\), however, must still be consistent. If it is estimated from predictive paths, consistency can be achieved by increasing the path length \(M\), increasing the number of parallel paths \(B\), or both.

\paragraph{Specialization to the vMF model.}
For the vMF model, the Bernstein--von Mises result applies directly in the
natural parameter \(c=\kappa\mu\in\mathbb R^p\). 
Suppose
\(X_1,\ldots,X_n\) are i.i.d. from a vMF distribution with finite natural
parameter \(c_\ast\). The parameter space is \(\mathbb R^p\), so every finite
\(c_\ast\) is an interior point. The log-likelihood has the exponential-family
form \(\ell_n(c)=n\{c^\top\overline X_n-\psi(c)\}\), with
\(\nabla\psi(c)=\mathbb E_c(X)\) and \(\nabla^2\psi(c)=I(c)\). Standard
regular likelihood theory for full exponential families therefore gives
\[
    c_{n,0}\overset{p}{\longrightarrow}c_\ast,
    \qquad
    \sqrt n(c_{n,0}-c_\ast)
    \rightsquigarrow
    \mathcal N_p\{0,I(c_\ast)^{-1}\}.
\]

The local regularity assumptions are also satisfied. By Corollary
\ref{cor:vmf_bound}, the score is globally bounded by \(2\), and \(I(c)\) is
Lipschitz and uniformly nonsingular on every bounded set
\(K_R=\{c:\|c\|\le R\}\). If \(R>\|c_\ast\|\), then
\(c_{n,0}\in K_R\) with probability tending to one. Moreover, since the
predictive displacement is a bounded-increment martingale,
\[
    \mathbb E\!\left(
        \sup_{m\ge0}\|c_{n,m}-c_{n,0}\|^2
        \,\middle|\, X_{1:n}
    \right)
    \le
    C\sum\nolimits_{m\ge1}(n+m)^{-2} 
    =
    O(n^{-1}),
\]
by Doob's inequality. Hence the probability that the predictive path leaves
\(K_R\) tends to zero for any sufficiently large \(R\). Thus localization is
asymptotically harmless for finite-concentration vMF models.

The remaining condition is consistency of the terminal information estimator.
For vMF models, one may use the analytic plug-in estimator
\(\widehat I_n=I(c_{n,0})\), in which case \(\Delta_{I,n}=0\). Alternatively,
the pooled quadratic-variation estimator satisfies
\(\|\widehat I^{\mathrm{pool}}_{n,M,B}-I(c_{n,0})\|_{\mathrm{op}}
=
O_p\{(\log p/(B N_{\mathrm{eff}}(n,M)))^{1/2}+\sqrt{w_{n,M}}\}\).
Thus it is consistent whenever \(B N_{\mathrm{eff}}(n,M)\to\infty\) and
\(w_{n,M}\to0\); in particular, for fixed \(M\), it is enough that
\(B\to\infty\).

Consequently, the Hybrid MPS-R draw
$
    c_n^{\mathrm{hyb,pc}}
    =
    c_{n,0}
    +
    \widehat I_n^{-1}\{D_{n,M_n}+Z_{n,M_n}\},
$
where 
$
Z_{n,M_n}\mid\mathcal F_{n,M_n}
\sim
\mathcal N_p(0,r_{n,M_n}\widehat I_n)$,
satisfies  
\[
    \mathcal L\!\left[
        \sqrt n\{c_n^{\mathrm{hyb,pc}}-c_{n,0}\}
        \,\middle|\, X_{1:n}
    \right]
    \rightsquigarrow
    \mathcal N_p\{0,I(c_\ast)^{-1}\}
\]
in probability, 
provided \(M_n/n\to0\) and \(\widehat I_n\) is consistent.

\paragraph{Implications for \(\kappa\) and mean direction.}
The theorem is stated for the Euclidean natural parameter \(c\). Smooth
functionals then follow by the delta method. If \(c_\ast\ne0\), then
\(\kappa(c)=\|c\|\) is differentiable at \(c_\ast\), and for \(p=2\) the mean
direction \(\phi(c)=\operatorname{atan2}(c_2,c_1)\) is locally smooth after
choosing a local angular branch. Hence posterior calibration for
\((\kappa,\phi)\) follows from the natural-parameter Bernstein--von Mises
limit. If \(c_\ast=0\), the natural-parameter limit remains valid, but
\(\kappa=\|c\|\) is not differentiable and the mean direction is not identified;
in that case, the Bernstein--von Mises statement should be interpreted in the
natural parameter \(c\).

\paragraph{Practical implications for choosing \(M\) and \(B\).}
Theorem \ref{thm:bvm} shows that \(M\) need not grow proportionally to \(n\). 
In regular low-dimensional models, a moderate fixed \(M\) can be sufficient, while additional computation is often better spent increasing \(B\), which exploits the embarrassingly parallel structure of the sampler.

For the vMF model, the globally bounded score and locally smooth Fisher information make the hybrid approximation especially stable on bounded concentration ranges. In the simulations in Section \ref{sec:simulation}, a relatively small truncation depth, such as \(M=100\), is already sufficient to match the MCMC benchmark closely when paired with a large number of parallel predictive paths.

\section{Simulation Studies}
\label{sec:simulation}


This section evaluates the proposed samplers against MCMC benchmarks for the von Mises--Fisher model. The experiments examine the effect of finite-horizon truncation, the correction supplied by the Gaussian tail, and the frequentist calibration of nominal credible sets. We consider circular data ($p=2$) and spherical data ($p=3$).

In both settings the true concentration is $\kappa^*=4$. For $p=2$ the true mean direction is $\phi^*=2$ radians, and the samplers operate on the natural parameter $c=\kappa(\cos\phi,\sin\phi)^\top\in\mathbb R^2$. For $p=3$ the true mean direction is
$\mu^*=\bigl(\sqrt6/4,\,\sqrt6/4,\,1/2\bigr)^\top$
(colatitude $\vartheta_1^*=\pi/3$, longitude $\vartheta_2^*=\pi/4$)
and the samplers operate on $c=\kappa\mu\in\mathbb R^3$.


\paragraph{Initialization.} Given observations $X_{1:n}\subset\mathbb S^{p-1}$, all samplers start from the maximum likelihood estimator $c_{n,0}$. Let $\bar X_n=n^{-1}\sum_{i=1}^n X_i$ and $R_n=\|\bar X_n\|$. When $0<R_n<1$, set $\widehat\mu_n=\bar X_n/R_n$, $\widehat\kappa_n=A_p^{-1}(R_n)$, and $c_{n,0}=\widehat\kappa_n\widehat\mu_n$. If $R_n$ is numerically indistinguishable from zero, we set $c_{n,0}=0$, corresponding to an approximately uniform initial fit.

\subsection{Benchmarks and Competing Samplers}
\label{sec:benchmarks}
For $p=2$ we use the conjugate Gibbs sampler of \citet{forbes2015fast}, in which both full conditionals are available in closed form, and we take the neutral hyperparameters $c_0=\mathcal R_0=\phi_0=0$. This construction relies on the closed form of the circular normalizing constant and does not extend to $p=3$. For the spherical case we therefore use the conjugate framework of \citet{straub2017nonparametric}, in which the full conditional for $\mu$ remains von Mises--Fisher while $\kappa$ is updated by a slice sampler, and we take the improper limiting prior obtained as $a,b\to0^{+}$. Neither benchmark carries prior information beyond what is needed for a proper posterior. Both are described in detail in Appendix~\ref{app:mcmc}.

Alongside the benchmarks we report three martingale posterior samplers. MPS-R is the raw-score sampler with terminal post-correction and no tail term, while Hybrid MPS-R adds the Gaussian tail. Hybrid MPS-FP replaces the raw-score increments in the predictive path by Fisher-preconditioned increments while retaining the same Gaussian tail, so that a comparison between the two hybrids isolates the cost of pathwise preconditioning. We use the hybrid rather than the fully truncated preconditioned scheme, since the latter requires $M\gg n$ to reach the same calibration and is correspondingly more expensive.

\subsection{Posterior Trajectories and the Effect of Truncation}\label{sec:single_sim}

We first examine the qualitative behavior of the marginal posterior for the concentration parameter $\kappa$ in the circular model across varying regimes of sample size $n$ and simulation depth $M$. Throughout these experiments, the ensemble size is fixed at $B=1000$ to ensure a stable pooled quadratic-variation estimate $\hat{I}_{\text{pool}}$.

At $n=10$ the martingale posterior retains the non-Gaussian shape of the likelihood, and both samplers preserve the right skewness of the posterior for $\kappa$. Here $M=1000$ is large relative to $n$, so the tail weight omitted by MPS-R is small, and the Gaussian tail allows Hybrid MPS-R to reproduce the same density from a path ten times shorter; both densities are close to the MCMC benchmark, as shown in Figure~\ref{fig:MPS_n10}. This is consistent with the finite-horizon bound of Theorem~\ref{thm:hybrid_bound}, 
under which the tail proxy restores the variance weight without altering the non-Gaussian features resolved by the early predictive steps. 



\begin{figure}[H]
    \centering
    \includegraphics[width=0.9\textwidth]{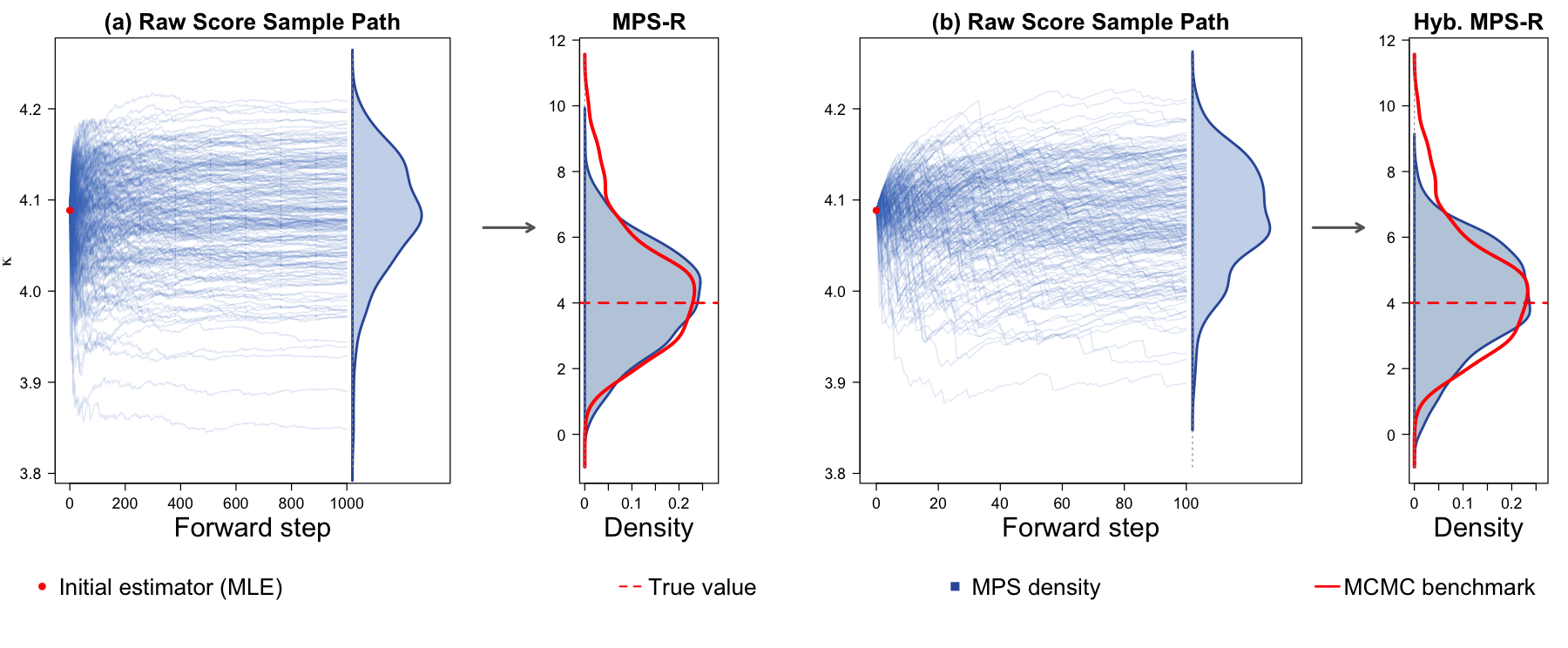}
    \captionsetup{width=0.90\linewidth, font=small}
    \caption{Marginal posterior density of $\kappa$: (a) $n=10$, MPS-R, $M=1000$, (b) $n=10$, Hybrid MPS-R, $M=100$.}

    \label{fig:MPS_n10}
\end{figure}



At $n=500$, the share of the total predictive variance $w_{n,\infty}\approx n^{-1}$ that an explicitly simulated path of length $M$ accumulates is $w_{n,M}\approx M/\{n(n+M)\}$, so even $M=1000$ leaves about one third of the required variance unresolved at this sample size. The MPS-R density is correspondingly under-dispersed relative to the benchmark, as seen in panel (a) of Figure~\ref{fig:MPS_n500}, and this is the finite-horizon variance deficit associated with the omitted tail weight $r_{n,M}$ in Section~\ref{sec:bvm}.  Hybrid MPS-R restores that weight by construction and agrees closely with the benchmark at $M=100$; see panel (b). The hybrid sampler is thus accurate at a horizon ten times shorter than the one at which MPS-R is already visibly deficient, which indicates that its behaviour is largely insensitive to the choice of simulation depth. The later coverage study in Section~\ref{sec:coverage} quantifies both effects.

\begin{figure}[H]
    \centering
    \includegraphics[width=0.9\textwidth]{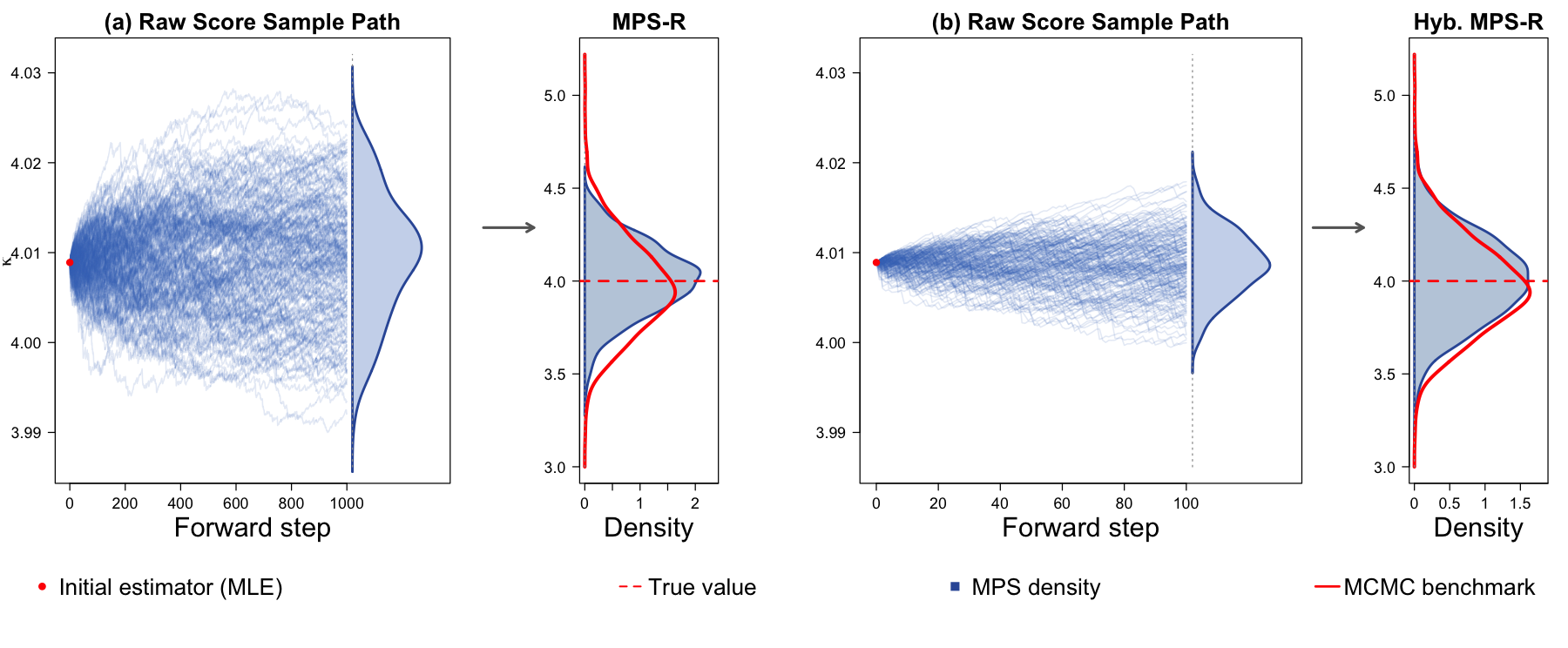}
    \captionsetup{width=0.90\linewidth, font=small}
    \caption{Marginal posterior density of $\kappa$: (a) $n=500$, MPS-R, $M=1000$, (b) $n=500$, Hybrid MPS-R, $M=100$.}
    \label{fig:MPS_n500}
\end{figure}

\subsection{Frequentist Coverage and Results}
\label{sec:coverage}

We now assess frequentist calibration over $5000$ independent replicates, using $B=500$ predictive paths, 
with simulation depth $M=1000$ for the truncated schemes and $M=100$ for the hybrid schemes. Each design is evaluated at the three sample sizes $n=10$, $100$ and $500$. We report the circular model first, followed by the spherical model.

\begin{table}[H]
\centering
\footnotesize
\captionsetup{width=0.90\linewidth, font=small}
\caption{Circular model ($p=2$). Empirical coverage and average length of nominal $95\%$ credible intervals over $5000$ replicates.}
\label{tab:sim_results_p2}
\begin{tabular}{l l cc cc cc}
\toprule
& & \multicolumn{2}{c}{$n = 10$} & \multicolumn{2}{c}{$n = 100$} & \multicolumn{2}{c}{$n = 500$} \\
\cmidrule(lr){3-4} \cmidrule(lr){5-6} \cmidrule(lr){7-8}
\textbf{Method} & \textbf{Param.} & Cov. (\%) & Len. & Cov. (\%) & Len. & Cov. (\%) & Len. \\
\midrule
\textbf{MCMC} & $\phi$ & 92.1 & 0.681 & 94.4 & 0.211 & 95.1 & 0.0941 \\
& $\kappa$ & 92.8 & 9.588 & 96.1 & 2.220 & 95.1 & 0.9734 \\
\textit{Normalized CPU time} & & \multicolumn{2}{c}{0.495 s} & \multicolumn{2}{c}{0.444 s} & \multicolumn{2}{c}{0.439 s} \\
\midrule
\textbf{Hybrid MPS-FP} & $\phi$ & 94.6 & 0.726 & 94.5 & 0.211 & 95.2 & 0.0938 \\
& $\kappa$ & 95.4 & 8.988 & 94.9 & 2.036 & 93.7 & 0.8909 \\
\textit{Normalized CPU time} & & \multicolumn{2}{c}{3.214 s} & \multicolumn{2}{c}{3.330 s} & \multicolumn{2}{c}{3.180 s} \\
\midrule
\textbf{MPS-R} & $\phi$ & 99.1 & 1.119 & 94.0 & 0.206 & 89.1 & 0.0767 \\
& $\kappa$ & 94.5 & 8.004 & 93.4 & 1.930 & 88.1 & 0.7260 \\
\textit{Normalized CPU time} & & \multicolumn{2}{c}{20.061 s} & \multicolumn{2}{c}{20.336 s} & \multicolumn{2}{c}{19.810 s} \\
\midrule
\textbf{Hybrid MPS-R} & $\phi$ & \textbf{99.2} & \textbf{1.122} & \textbf{95.0} & \textbf{0.215} & \textbf{95.1} & \textbf{0.0940} \\
& $\kappa$ & \textbf{94.8} & \textbf{8.054} & \textbf{94.7} & \textbf{2.027} & \textbf{93.6} & \textbf{0.8902} \\
\textit{Normalized CPU time} & & \multicolumn{2}{c}{2.801 s} & \multicolumn{2}{c}{2.791 s} & \multicolumn{2}{c}{2.771 s} \\
\bottomrule
\multicolumn{8}{p{13cm}}{\footnotesize \textit{Note:} Truncated schemes use $M=1000$ and hybrid schemes use $M=100$, with $B=500$ predictive paths. For MPS, runtime is single-core seconds per $1000$ generated output draws; for MCMC, it is seconds per $1000$ estimated effective draws after a $5000$-iteration burn-in. Under pooled calibration, MPS outputs share an estimated calibration matrix and are therefore not strictly independent. For coverage evaluation, angular draws are unwrapped on the branch centered at the true $\phi^*$ before equal-tailed intervals are formed.}
\end{tabular}
\end{table}

In the circular model, Hybrid MPS-R covers $\kappa$ close to the nominal level at all three sample sizes, between $93.6\%$ and $94.8\%$, and covers $\phi$ at $95.0\%$ and $95.1\%$ at $n=100$ and $n=500$; see Table~\ref{tab:sim_results_p2}. Its interval lengths sit at the inverse-Fisher scale, for instance $0.890$ for $\kappa$ at $n=500$ against the asymptotic value $2z_{0.975}\{nA_p'(\kappa^*)\}^{-1/2}\approx0.89$. This is consistent with the mechanism described by Theorem~\ref{thm:bvm}, under which the hybrid draw attains the Bernstein--von Mises limit at fixed simulation depth provided the terminal calibration matrix is consistent. Here that condition is approximated numerically with a pooled estimate based on $B=500$ paths; formally, for fixed $M$, the consistency statement corresponds to $B\to\infty$. At $n=10$, the Hybrid MPS-R interval for $\phi$ is conservative, with $99.2\%$ coverage and average length $1.122$, compared with $92.1\%$ and $0.681$ for MCMC. Thus terminal covariance calibration does not imply exact finite-sample distributional agreement. The benchmark intervals for $\kappa$ are somewhat wider than the inverse-Fisher scale at $n=100$, which is consistent with their coverage of $96.1\%$, while at $n=10$ the benchmark covers $\kappa$ at $92.8\%$, slightly further from the nominal level than Hybrid MPS-R.

MPS-R exhibits the deterministic variance shrinkage described in Section~\ref{sec:bvm}.
The post-corrected truncated draw retains only the realized weight $w_{n,M}$, so to first order its spread is smaller than the hybrid spread by the factor $\{w_{n,M}/w_{n,\infty}\}^{1/2}$. At simulation depth $M=1000$ this factor equals $0.995$, $0.953$ and
$0.816$ for $n=10$, $100$ and $500$.
The ratios of MPS-R to Hybrid MPS-R interval lengths in
Table~\ref{tab:sim_results_p2} match these values for both $\phi$ and $\kappa$.
Coverage of MPS-R accordingly falls to $88$--$89\%$ at $n=500$, where the omitted tail $r_{n,M}$ is largest relative to $w_{n,\infty}$.
The Gaussian tail restores the missing deterministic weight without lengthening the explicit path; Theorem~\ref{thm:bvm} then gives Bernstein--von Mises calibration for the hybrid sampler, which truncated MPS-R cannot attain unless $M\gg n$.

We next repeat the study on the sphere, where the benchmark is the sampler of \citet{straub2017nonparametric} and the directional target is the mean direction $\mu$ rather than a single angle. The results of Section~\ref{sec:bvm} are stated for a regular parametric model on $\mathbb{R}^p$, and the shrinkage factor $\{w_{n,M}/w_{n,\infty}\}^{1/2}$ depends only on the sample size and the simulation depth, so the conclusions drawn for the circular model should carry over to $p=3$ without modification. This is what the spherical design shows; see Table~\ref{tab:sim_results_p3}. Hybrid MPS-R covers $\mu$ and $\kappa$ between $94.7\%$ and $95.3\%$ at $n=100$ and $n=500$, with average credible-set sizes within $1.5\%$ of the benchmark, and Hybrid MPS-FP is calibrated over the same range. Coverage of MPS-R falls to $86.8\%$ for $\mu$ and $88.9\%$ for $\kappa$ at $n=500$, with length ratios of $0.816$ and $0.815$ relative to Hybrid MPS-R, matching the factor reported above.

\begin{table}[H]
\centering
\footnotesize
\captionsetup{width=0.90\linewidth, font=small}
\caption{Spherical model ($p=3$). Empirical coverage and average size of nominal $95\%$ credible sets over $5000$ replicates.}
\label{tab:sim_results_p3}
\begin{tabular}{l l cc cc cc}
\toprule
& & \multicolumn{2}{c}{$n = 10$} & \multicolumn{2}{c}{$n = 100$} & \multicolumn{2}{c}{$n = 500$} \\
\cmidrule(lr){3-4} \cmidrule(lr){5-6} \cmidrule(lr){7-8}
\textbf{Method} & \textbf{Param.} & Cov. (\%) & Len. & Cov. (\%) & Len. & Cov. (\%) & Len. \\
\midrule
\textbf{MCMC} & $\mu$ & 93.9 & 26.603 & 95.1 & 8.077 & 94.9 & 3.6043 \\
& $\kappa$ & 92.7 & 6.248 & 94.6 & 1.602 & 94.7 & 0.7048 \\
\textit{Normalized CPU time} & & \multicolumn{2}{c}{0.946 s} & \multicolumn{2}{c}{0.835 s} & \multicolumn{2}{c}{0.827 s} \\
\midrule
\textbf{Hybrid MPS-FP} & $\mu$ & 93.8 & 25.758 & 95.1 & 8.083 & 95.1 & 3.6104 \\
& $\kappa$ & 92.8 & 6.150 & 94.7 & 1.600 & 94.6 & 0.7043 \\
\textit{Normalized CPU time} & & \multicolumn{2}{c}{5.622 s} & \multicolumn{2}{c}{5.283 s} & \multicolumn{2}{c}{5.382 s} \\
\midrule
\textbf{MPS-R} & $\mu$ & 97.7 & 30.984 & 94.1 & 7.822 & 86.8 & 2.9492 \\
& $\kappa$ & 95.4 & 5.713 & 93.5 & 1.515 & 88.9 & 0.5741 \\
\textit{Normalized CPU time} & & \multicolumn{2}{c}{31.383 s} & \multicolumn{2}{c}{32.484 s} & \multicolumn{2}{c}{29.203 s} \\
\midrule
\textbf{Hybrid MPS-R} & $\mu$ & \textbf{97.7} & \textbf{31.040} & \textbf{95.3} & \textbf{8.186} & \textbf{95.0} & \textbf{3.6156} \\
& $\kappa$ & \textbf{95.8} & \textbf{5.742} & \textbf{94.9} & \textbf{1.592} & \textbf{94.7} & \textbf{0.7046} \\
\textit{Normalized CPU time} & & \multicolumn{2}{c}{4.537 s} & \multicolumn{2}{c}{5.062 s} & \multicolumn{2}{c}{3.882 s} \\
\bottomrule
\multicolumn{8}{p{13cm}}{\footnotesize \textit{Note:} For $\kappa$, Len.\ is the average length of the equal-tailed credible interval; for $\mu$, it is the average angular radius in degrees of the credible region. Other settings are as in Table~\ref{tab:sim_results_p2}.}
\end{tabular}
\end{table}

Both benchmarks are efficient for the vMF family. The circular Gibbs sampler exploits a closed-form conjugate structure and the spherical sampler exploits the low-dimensional geometry of the three-dimensional model, and in single-core terms both are faster than the martingale posterior samplers. The predictive paths can be generated independently and require one synchronization step when pooled calibration is used, so this path-generation stage is parallelizable; the final pooled outputs share the estimated calibration matrix. Among the martingale posterior samplers, the hybrid schemes are considerably cheaper than MPS-R because they use a shorter explicit path, and Hybrid MPS-R is cheaper than Hybrid MPS-FP in both dimensions, with the gap widening from $p=2$ to $p=3$. The second comparison is the more relevant one for the richer directional models discussed in Section~\ref{sec:conclusion}, where the Fisher information is less readily available and pathwise preconditioning is correspondingly more costly.

\section{Real-data application: OSCAR surface currents}
\label{sec:application}
We apply the proposed martingale posterior sampler to satellite-derived ocean surface-current data from the OSCAR L4 product \citep{https://doi.org/10.5067/oscar-25i20}. The data provide gridded eastward and northward near-surface velocity components, denoted by $(u,v)$. At each valid ocean grid cell, we convert the velocity vector into a current direction $\omega=\operatorname{atan2}(v,u)$, giving circular observations on the unit circle. The analysis focuses on the California Current System (CCS), an eastern-boundary current along the west coast of North America. We consider the regional window
\[
30^\circ\mathrm{N}\le {\rm lat}\le 45^\circ\mathrm{N},
\qquad
115^\circ\mathrm{W}\le {\rm lon}\le 130^\circ\mathrm{W}.
\]
To provide a snapshot-level comparison of local directional structure, we analyze one summer date, 1 August 2024, and one winter date, 15 January 2024. These two dates are used to illustrate a summer--winter contrast within the region, rather than to estimate a full seasonal climatology.

All current directions are expressed relative to a common CCS reference orientation, $\omega_{\rm ref}=-67.5^\circ$, corresponding to a south-southeast, equatorward alongshore direction. Using the same reference orientation for both snapshots avoids introducing date-specific shifts through separate recentering. For each valid grid cell $x$, we define
\[
    \Delta\omega(x)
    =
    \operatorname{wrap}\{\omega(x)-\omega_{\rm ref}\}
    \in (-\pi,\pi],
    \qquad
    Y(x)=\{\cos\Delta\omega(x),\sin\Delta\omega(x)\}\in S^1 .
\]
The transformed observations are therefore unit vectors and are summarized using a two-dimensional vMF working model.

The full CCS window is spatially heterogeneous, and a single unimodal vMF distribution is not intended to describe the entire region. We therefore restrict the formal vMF analysis to a compact southern offshore subarea,
\[
    30.50^\circ\mathrm{N}\leq {\rm lat}\leq 32.25^\circ\mathrm{N},
    \qquad
    126.50^\circ\mathrm{W} \leq {\rm lon} \leq 128.25^\circ\mathrm{W}.
\]
This subarea contains $N=64$ common valid grid cells in both snapshots. It is sufficiently localized for a single-vMF summary to be interpretable, while still showing a visible date-to-date contrast in current orientation. The full CCS window and the selected subarea are shown in Figure~\ref{fig:oscar_compact_map}.

\begin{figure}[H]
  \centering
  \includegraphics[width=0.90\linewidth]{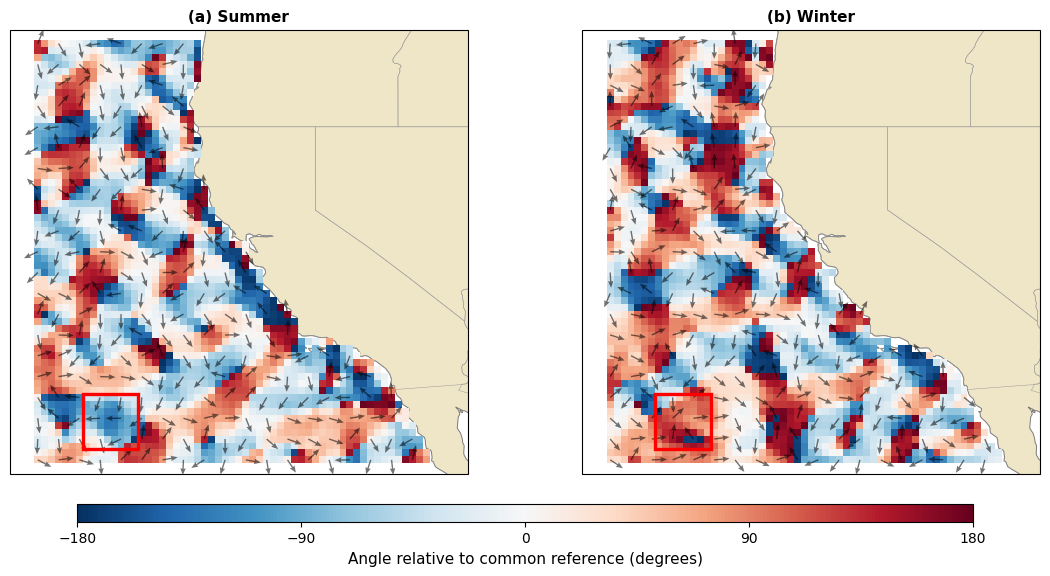}
  \captionsetup{width=0.90\linewidth, font=small}
\caption{OSCAR current directions in the California Current System window for (a) the summer snapshot and (b) the winter snapshot. Colours show angular deviation from the common reference direction $\omega_{\rm ref}=-67.5^\circ$, arrows show normalized local current directions, and the red rectangle marks the southern offshore subarea used for the single-vMF analysis.}
  \label{fig:oscar_compact_map}
\end{figure}


Within the selected subarea, the exploratory summaries suggest that the difference between the two snapshots is mainly associated with the dominant orientation, rather than with a large change in directional concentration. In both snapshots, the directions are moderately concentrated, but the summer and winter samples are centred around different angular locations relative to the common CCS reference orientation. This pattern is visible both in the transformed unit vectors and in the overlaid rose histograms in Figure~\ref{fig:oscar_compact_raw}.

\begin{figure}[H]
  \centering
  \includegraphics[width=0.9\linewidth]{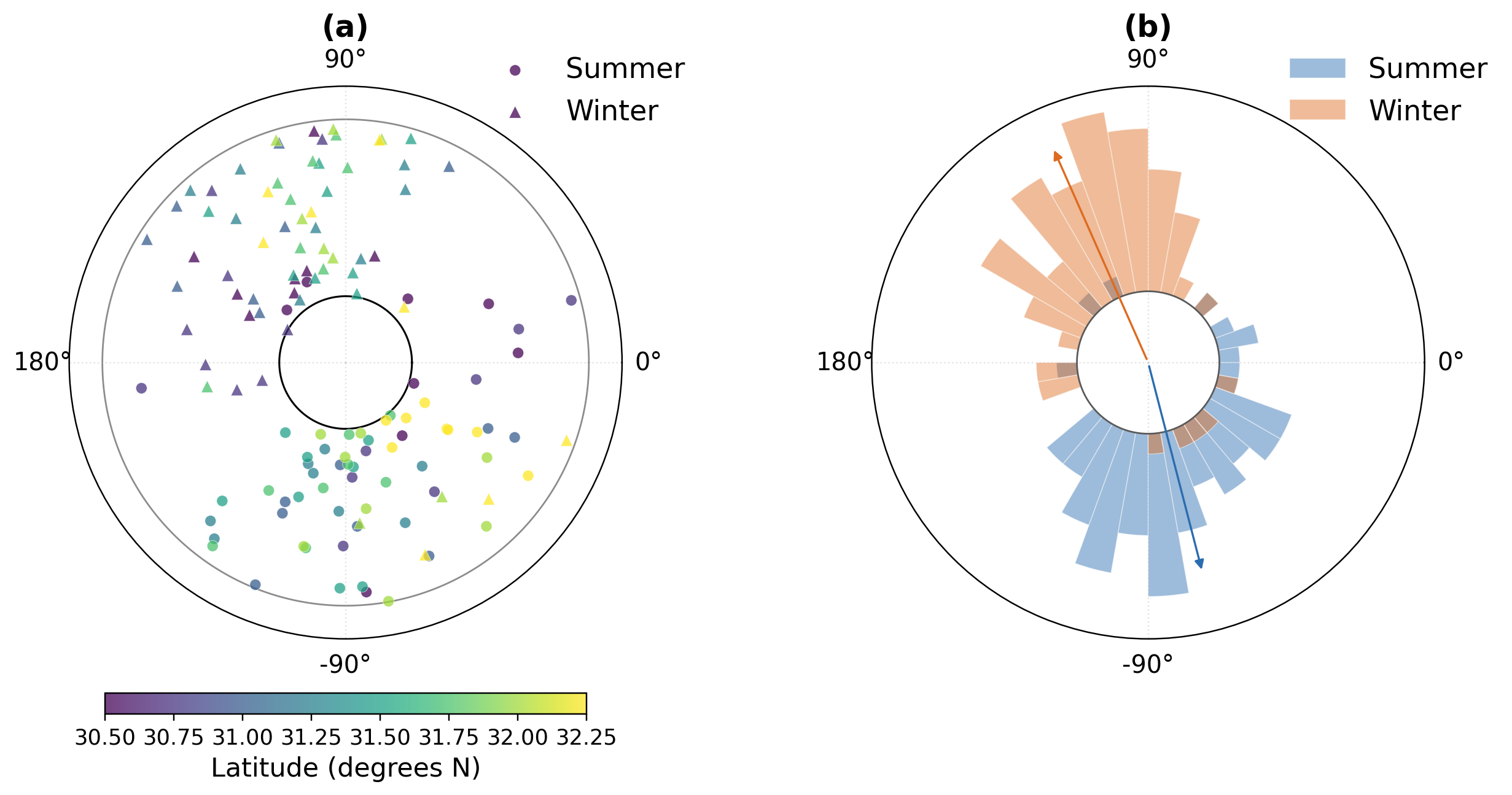}
  \captionsetup{width=0.90\linewidth, font=small}
\caption{Raw directional observations in the selected southern offshore subarea. Panel (a) shows the transformed unit vectors $Y(x)$ on the circle, with colour indicating latitude and marker type distinguishing the summer and winter snapshots. Panel (b) shows overlaid rose histograms of angular deviations from the common reference direction. Blue denotes summer and orange denotes winter.}

  \label{fig:oscar_compact_raw}
\end{figure}

We fit separate two-dimensional vMF working models to the summer and winter samples. The inferential targets are the concentration parameter $\kappa$ and the angular mean direction $\phi = \operatorname{atan2}(\mu_2,\mu_1)$, where $\phi$ is reported in degrees relative to the common CCS reference direction. All samplers operate on the natural parameter $c=\kappa\mu$ and are initialized at the MLE. We use Hybrid MPS-R with $M=120$ predictive steps and $B=800$ paths, since the Gaussian tail supplies the variance weight omitted by truncation and a short explicit path is therefore sufficient. For comparison we report MPS-R, which carries no tail correction and is run to $M=1200$ so that the residual weight is small, Hybrid MPS-FP, and the conjugate vMF MCMC posterior used in the simulation study with neutral hyperparameters $R_0=0$ and $c_0=0$. The reported credible intervals are model-based summaries under an independent single-vMF working model and do not account for spatial dependence among nearby grid cells.

Within the selected subarea, the two snapshots differ more in mean orientation than in directional concentration. Under Hybrid MPS-R, the posterior mean of $\kappa$ is 2.031 in summer and 2.117 in winter, with overlapping 95\% credible intervals, whereas the posterior mean of $\phi$ is $-75.4^\circ$ in summer and $113.8^\circ$ in winter, with well separated intervals; see Table~\ref{tab:oscar_posteriors}. The fitted mean deviations are separated by approximately $171^\circ$, while the local flow remains moderately concentrated in both snapshots. This illustrates a pronounced between-date directional contrast; two snapshots from a selected offshore subarea do not by themselves establish a seasonal reversal.

The four samplers give similar summaries. Posterior means agree to within $0.02$ for $\kappa$ and $0.6^\circ$ for $\phi$, and the credible intervals for $\phi$ have nearly identical widths under all four methods, while the martingale posterior intervals for $\kappa$ are somewhat narrower than the neutral-prior MCMC intervals, as also observed at comparable sample sizes in Section~\ref{sec:coverage}. The Gibbs sampler is the fastest in serial CPU time, requiring about $0.35$ s per $1000$ estimated effective draws, against about $3.7$ s per $1000$ generated draws for Hybrid MPS-R and $34$ s for MPS-R, whose explicit path is ten times longer. The predictive paths can be generated independently before the pooled calibration step, so this stage can be parallelized; Section~\ref{sec:simulation} discusses the comparison in more detail.

\begin{table}[H]
\centering
\small
\captionsetup{width=0.90\linewidth, font=small}
\caption{Posterior summaries for the single-vMF analysis of the selected southern offshore OSCAR subarea.}
\label{tab:oscar_posteriors}
\setlength{\tabcolsep}{5pt}

\begin{tabular}{llcccc}
\toprule
 & & \multicolumn{2}{c}{Summer, $N=64$}
   & \multicolumn{2}{c}{Winter, $N=64$} \\
\cmidrule(lr){3-4} \cmidrule(lr){5-6}
Method & Param. & Mean & 95\% CrI & Mean & 95\% CrI \\
\midrule

\multirow{2}{*}{\textbf{MCMC}}
 & $\kappa$ & 2.009 & $[1.354,\,2.831]$
 & 2.110 & $[1.428,\,2.900]$ \\
 & $\phi$ & -75.6 & $[-87.8,\,-63.7]$
 & 113.7 & $[103.1,\,125.8]$ \\
\textit{Normalized CPU time}
 & & \multicolumn{2}{c}{0.374 s}
 & \multicolumn{2}{c}{0.325 s} \\



\midrule

\multirow{2}{*}{\textbf{Hybrid MPS-FP}}
 & $\kappa$ & 2.031 & $[1.482,\,2.669]$
 & 2.099 & $[1.505,\,2.791]$ \\
 & $\phi$ & -75.2 & $[-87.0,\,-64.1]$
 & 113.8 & $[101.8,\,124.8]$ \\
\textit{Normalized CPU time}
 & & \multicolumn{2}{c}{4.986 s}
 & \multicolumn{2}{c}{5.038 s} \\

\midrule

\multirow{2}{*}{\textbf{MPS-R}}
 & $\kappa$ & 2.039 & $[1.336,\,2.611]$
 & 2.097 & $[1.432,\,2.694]$ \\
 & $\phi$ & -75.3 & $[-87.8,\,-63.9]$
 & 113.5 & $[101.5,\,125.1]$ \\
\textit{Normalized CPU time}
 & & \multicolumn{2}{c}{33.766 s}
 & \multicolumn{2}{c}{35.208 s} \\

\midrule

\multirow{2}{*}{\textbf{Hybrid MPS-R}}
 & \textbf{$\kappa$} & \textbf{2.031} & $\textbf{[1.420,\,2.585]}$
 & \textbf{2.117} & $\textbf{[1.456,\,2.730]}$ \\
 & \textbf{$\phi$} & \textbf{-75.4} & $\textbf{[-86.9,\,-63.0]}$
 & \textbf{113.8} & $\textbf{[101.2,\,124.9]}$ \\
\textit{Normalized CPU time}
 & & \multicolumn{2}{c}{3.683 s}
 & \multicolumn{2}{c}{3.669 s} \\

\bottomrule
\end{tabular}

\medskip
\begin{minipage}{0.90\linewidth}
\footnotesize
\textit{Note:}
CrI denotes the equal-tailed 95\% credible interval.
The truncated schemes use $M=1200$, whereas the hybrid schemes use
$M=120$, with $B=800$ predictive paths.
For MPS, runtime is single-core seconds per 1000 generated output draws; for
MCMC, it is seconds per 1000 estimated effective draws after a
5000-iteration burn-in. Under pooled calibration, the MPS outputs share an
estimated calibration matrix and are therefore not strictly independent.
The corresponding MLEs are
$(\hat{\kappa},\hat{\phi})=(2.017,-75.5^\circ)$ for summer and
$(2.109,113.9^\circ)$ for winter.
\end{minipage}

\end{table}

The predictive paths show how the terminal correction produces this scale. Along the raw-score path the state moves only slightly from its initialization, with $\kappa$ remaining within roughly $1.9$ to $2.3$, because the raw increments accumulate variance on the information scale. The terminal correction and the Gaussian tail rescale the draw to the inverse-information scale, after which the terminal densities are close to the neutral-prior MCMC densities for both parameters and both snapshots; see Figure~\ref{fig:oscar_compact_traces}. The pathwise view shows the same separation, with the terminal distributions for $\kappa$ centred at similar values in the two snapshots and those for $\phi$ concentrating around distinct angular regions.

\begin{figure}[H]
  \centering
  \includegraphics[width=0.90\linewidth]{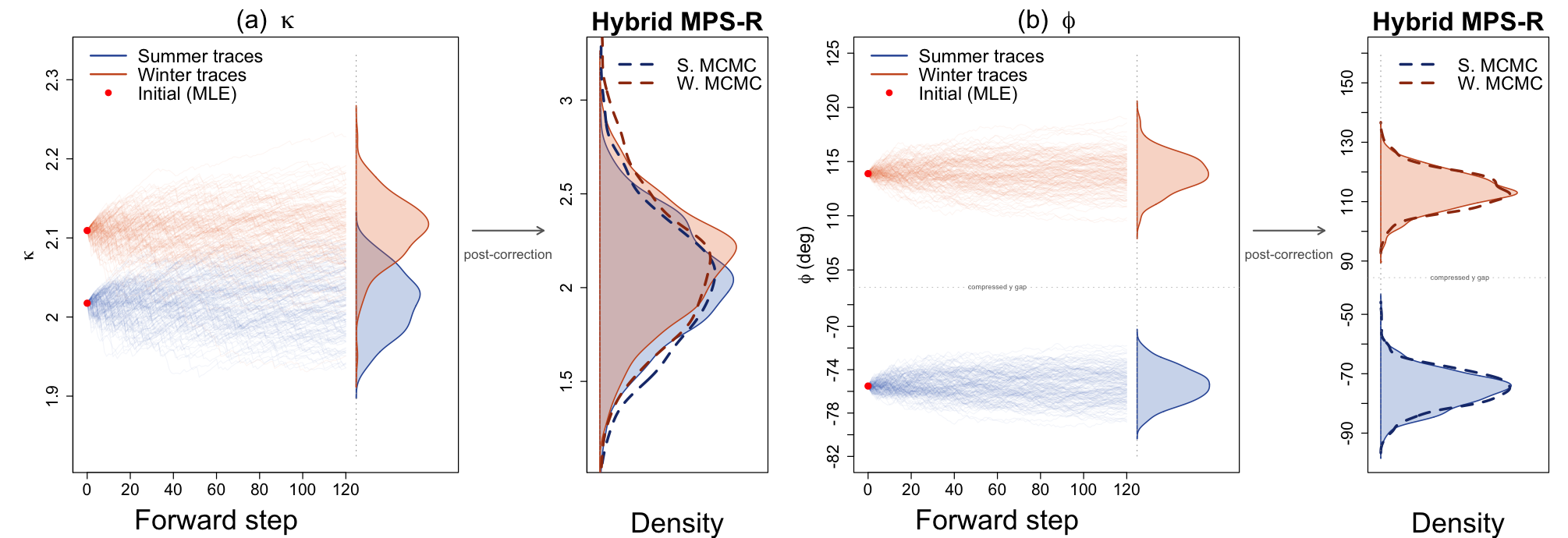}
  \captionsetup{width=0.90\linewidth, font=small}
  \caption{Hybrid MPS-R output and MCMC benchmark for the selected southern offshore OSCAR subarea. Panel (a) shows the concentration parameter $\kappa$ and panel (b) the mean direction $\phi$, in degrees relative to $\omega_{\rm ref}$. The left plots give the raw-score trajectories over $M=120$ steps from the MLE (red dot); the right plots give the terminal densities after post-correction, with the neutral-prior MCMC posteriors dashed. Blue denotes summer and orange denotes winter.}
  \label{fig:oscar_compact_traces}
\end{figure}

On the observation scale, the two fits have comparable concentration but mean directions that point to different angular regions, so the posterior bands have similar widths and are centred at different locations; see Figure~\ref{fig:oscar_compact_posterior_bands}. These summaries are consistent with Table~\ref{tab:oscar_posteriors} and Figure~\ref{fig:oscar_compact_traces}.

\begin{figure}[H]
  \centering
  \includegraphics[width=0.9\linewidth]{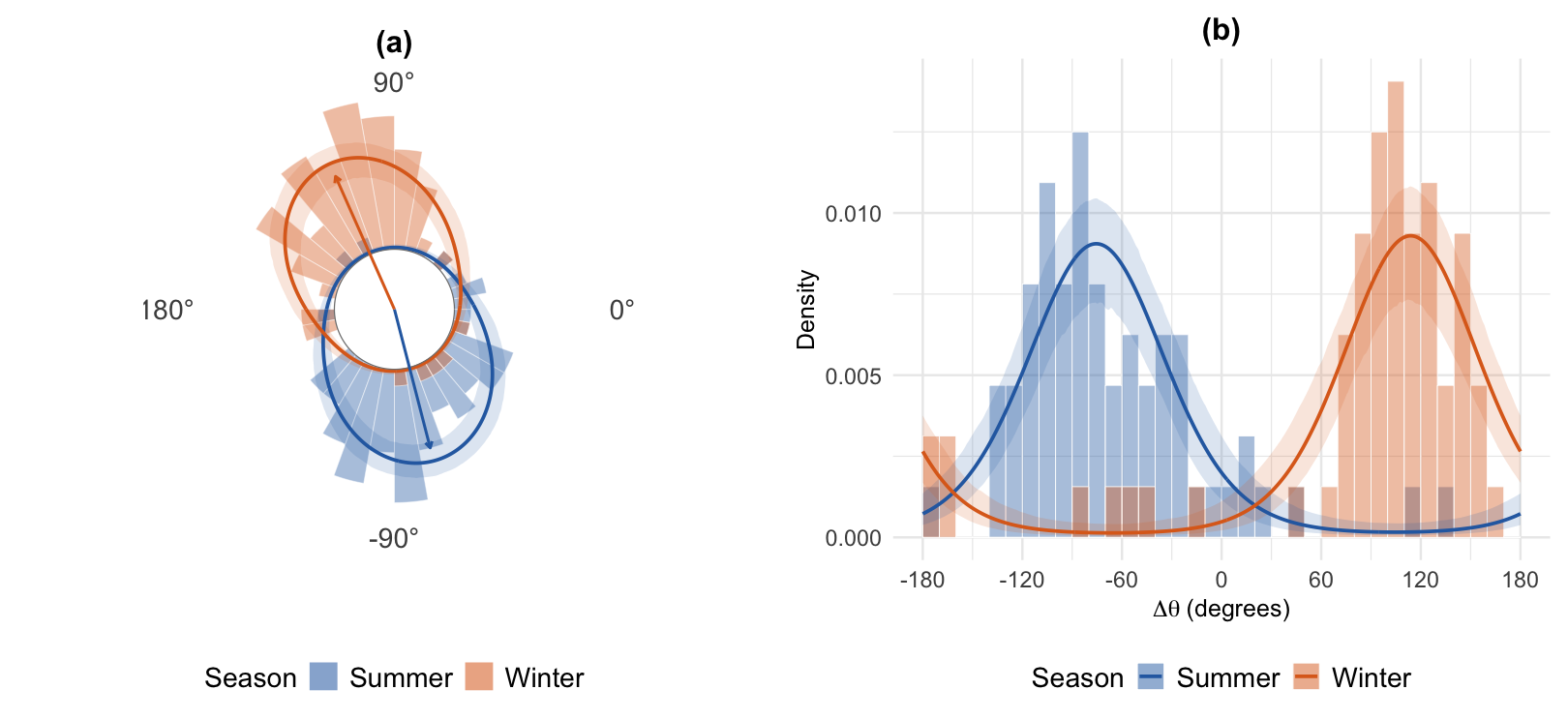}
  \captionsetup{width=0.90\linewidth, font=small}
  \caption{Hybrid MPS-R vMF fit for the selected southern offshore OSCAR subarea. Panel (a) overlays the seasonal rose histograms, fitted mean directions, fitted vMF densities, and $95\%$ posterior bands on the polar scale. Panel (b) shows the corresponding density view for $\Delta\omega$, in degrees relative to $\omega_{\rm ref}$.}
  \label{fig:oscar_compact_posterior_bands}
\end{figure}

Overall, this application illustrates martingale posterior uncertainty quantification for local directional summaries derived from a geophysical velocity field. In the selected southern offshore CCS subarea, the two OSCAR snapshots have comparable fitted directional concentration but different fitted mean orientations, and the analysis separates uncertainty in the concentration parameter $\kappa$ from uncertainty in the local mean direction $\phi$. These conclusions should be interpreted as model-based summaries for the selected subarea and dates, rather than as a full seasonal analysis of the California Current System.

\section{Concluding Remarks}
\label{sec:conclusion} 

This paper has studied a raw-score martingale posterior sampler for von Mises--Fisher models. The construction separates martingale simulation from covariance calibration: the predictive path uses only mean-zero score increments, while the inverse-information scale required for Bernstein--von Mises calibration is imposed through a terminal linear correction. A hybrid version further replaces the unresolved part of the infinite predictive continuation by a Gaussian tail with matching leading-order variance.

The vMF model provides a transparent setting for developing this construction, since its natural-parameter score is bounded and its Fisher information is explicit. This also means that the basic vMF model is not the setting in which avoiding pathwise Fisher preconditioning is most computationally consequential. Its role here is instead to provide a tractable benchmark in which the martingale property, finite-horizon variance loss, hybrid tail correction, and terminal covariance calibration can be examined directly.

The separation between predictive simulation and covariance calibration may be more useful in richer directional models. In finite or nonparametric mixtures of vMF components, information-based preconditioning may involve component allocations, mixture weights, cross-component dependence, and label structure. In rotational or frame-valued models, such as matrix Langevin models on the Stiefel manifold, normalizing constants, parameter constraints, and information geometry are more involved. In these settings, a score-only predictive path with terminal calibration may reduce the need for repeated analytic information calculations along the simulated path.

The real-data analysis should also be interpreted within the limits of the working model. The OSCAR application used an independent single-vMF summary on a localized subarea and two selected dates. The resulting credible intervals are therefore model-based uncertainty summaries for local directional structure, not a full spatial or seasonal analysis. Extending the predictive law to account for spatial dependence, temporal dependence, or latent regime structure is a natural direction for future work.

The proposed sampler is closely related to, but distinct from, classical maximum likelihood inference. We do not propose a new point estimator for the vMF natural parameter; throughout the construction, the predictive recursion is initialized at the MLE. In regular parametric models, Bernstein--von Mises theory implies that a well-calibrated posterior uncertainty distribution should be centered at an efficient estimator, such as the MLE, and should have inverse-Fisher covariance at first order. Thus, agreement with MLE asymptotics is a calibration requirement for the proposed uncertainty distribution. 

Several qualifications delimit this calibration claim. The raw-score recursion is defined in the chosen natural parameterization and is not invariant under nonlinear reparameterization, so finite-sample post-corrected distributions may depend on that choice even though smooth functionals inherit the first-order limit by the delta method. The current covariance argument also uses the model information identity and assumes correct specification; under misspecification, a sandwich-type correction involving both the score variance and expected Hessian would generally be required. Finally, when calibration is estimated from the same pooled paths used as output, the returned draws share a random matrix and are not strictly independent.

Overall, the proposed method provides a prior-free predictive sampler for uncertainty quantification. It generates predictive uncertainty through forward score updates, without requiring MCMC transitions or repeated numerical optimization along the simulated path. The terminal post-correction and, when needed, the hybrid Gaussian tail recover the inverse-Fisher covariance scale required by Bernstein--von Mises calibration. While the vMF model offers a tractable benchmark, the broader motivation is the possibility of using score-based predictive simulation with terminal covariance calibration in directional models where pathwise Fisher preconditioning is less convenient.

\newpage 
\bibliographystyle{chicago}
\bibliography{references}

\newpage 
\appendix
\section{Appendix}
\label{sec:appendix} 
\subsection{Auxiliary Results}
\label{app:auxiliary_results}

We first record two elementary approximation tools used repeatedly in the proofs. Throughout the Appendix, \(C_p\) denotes a finite positive constant depending only on the dimension \(p\). Its value may change from line to line.

\begin{lemma}[Gaussian covariance perturbation]
\label{lem:gaussian_cov_perturb}
Let \(G_0\sim N_p(0,\Sigma_0)\) and \(G_1\sim N_p(0,\Sigma_1)\), where \(\Sigma_0\) and \(\Sigma_1\) are symmetric nonnegative definite matrices. Then, for every \(\varphi\in C_b^2(\mathbb R^p)\),
\[
    \left|
    \mathbb E\{\varphi(G_1)\}
    -
    \mathbb E\{\varphi(G_0)\}
    \right|
    \le
    C_p
    \|D^2\varphi\|_\infty
    \|\Sigma_1-\Sigma_0\|_{\mathrm{op}} .
\]
\end{lemma}

\begin{proof}
Let \(\Sigma_t=(1-t)\Sigma_0+t\Sigma_1\), \(0\le t\le1\), and let \(G_t\sim N_p(0,\Sigma_t)\). By the Gaussian interpolation identity, \(\frac{d}{dt}\mathbb E\{\varphi(G_t)\}=\frac12\mathbb E[\operatorname{tr}\{(\Sigma_1-\Sigma_0)D^2\varphi(G_t)\}]\).
Hence
\[
\begin{aligned}
    \left|
    \mathbb E\{\varphi(G_1)\}
    -
    \mathbb E\{\varphi(G_0)\}
    \right|
    &\le
    \frac12
    \int_0^1
    \mathbb E
    \left|
        \operatorname{tr}
        \{(\Sigma_1-\Sigma_0)D^2\varphi(G_t)\}
    \right|
    dt  \\
    &\le
    C_p
    \|D^2\varphi\|_\infty
    \|\Sigma_1-\Sigma_0\|_{\mathrm{op}} .
\end{aligned}
\]
The result follows. If one of the covariance matrices is singular, the same argument follows by adding \(\eta I_p\) and sending \(\eta\downarrow0\).
\end{proof}

\begin{lemma}[Smooth Gaussian approximation for the predictive tail]
\label{lem:tail_gaussian}
Fix \(n\) and \(M\). Work conditionally on \(\mathcal F_{n,M}\), and suppose that, under this conditional predictive law, the continuation after time \(M\) remains almost surely in a compact set \(K\) on which Assumption \ref{assum:hybrid_regularity} holds. Let \(T_{n,M}=\sum_{m=M+1}^{\infty}\gamma_ms_{n,m}\), \(r_{n,M}=\sum_{m=M+1}^{\infty}\gamma_m^2\), and \(H_M=I(\theta_{n,M})\).
Let \(G_M\sim N_p(0,r_{n,M}H_M)\). Then, for every \(f\in C_b^3(\mathbb R^p)\),
\[
\begin{aligned}
    \left|
    \mathbb E\{f(T_{n,M})\mid \mathcal F_{n,M}\}
    -
    \mathbb E\{f(G_M)\mid \mathcal F_{n,M}\}
    \right|
    \le
    C_p
    \left\{
        \|D^2 f\|_\infty L_KS_K r_{n,M}^{3/2}
        +
        \|D^3 f\|_\infty S_K^3 A_{n,M}
    \right\},
\end{aligned}
\]
where \(A_{n,M}=\sum_{m=M+1}^{\infty}\gamma_{m}^3\).
\end{lemma}

\begin{proof}
For \(L>M\), define \(T_{n,M,L}=\sum_{m=M+1}^L \gamma_{m}s_{n,m}\) and
\(r_{n,M,L}=\sum_{m=M+1}^L \gamma_{m}^2\). A standard Lindeberg
replacement argument for martingale differences compares \(T_{n,M,L}\) with a
Gaussian vector having covariance \(r_{n,M,L}H_M\).
At the \(m\)-th replacement step, the difference between the conditional covariance of the martingale increment and the frozen covariance is \(\gamma_m^2\{I(\theta_{n,m-1})-H_M\}\). Taylor expansion to second order therefore bounds the cumulative covariance-freezing contribution by \(C_p\|D^2f\|_\infty\sum_{m=M+1}^L\gamma_m^2\mathbb E[\|I(\theta_{n,m-1})-H_M\|_{\mathrm{op}}\mid\mathcal F_{n,M}]\).
By Lipschitz continuity of \(I(\cdot)\) on \(K\),
\(\|I(\theta_{n,m-1})-H_M\|_{\mathrm{op}}
\le L_K\|\theta_{n,m-1}-\theta_{n,M}\|\). Moreover, since the score is
bounded by \(S_K\) on \(K\),
\[
\begin{aligned}
    \mathbb E
    \left[
        \|\theta_{n,m-1}-\theta_{n,M}\|^2
        \mid \mathcal F_{n,M}
    \right]
    &=
    \mathbb E
    \left[
        \left\|
        \sum_{j=M+1}^{m-1}
        \gamma_{j}s_{n,j}
        \right\|^2
        \mid \mathcal F_{n,M}
    \right]  \\
    &\le
    S_K^2
    \sum_{j=M+1}^{m-1}\gamma_{j}^2
    \le
    S_K^2 r_{n,M}.
\end{aligned}
\]
Thus \(\sum_{m=M+1}^L\gamma_m^2\mathbb E[\|I(\theta_{n,m-1})-H_M\|_{\mathrm{op}}\mid\mathcal F_{n,M}]\le L_KS_Kr_{n,M}^{3/2}\).

The third-order Taylor remainders are bounded by \(C_p\|D^3f\|_\infty\sum_{m=M+1}^L\gamma_m^3\mathbb E[\|s_{n,m}\|^3\mid\mathcal F_{n,m-1}]\le C_p\|D^3f\|_\infty S_K^3A_{n,M}\).
The analogous Gaussian replacement remainders satisfy the same bound because
\(\|H_M\|_{\mathrm{op}}\le S_K^2\). Letting \(L\to\infty\) gives the stated result by dominated convergence.
\end{proof}

\begin{proposition}[Consistency of the pooled quadratic-variation estimator]
\label{prop:pooled_information_consistency}
Suppose Assumption \ref{assum:hybrid_regularity} holds on a compact set \(K\), and suppose the first \(M\) predictive steps of all \(B\) chains remain in \(K\). Let \(\widehat I_{n,M,B}^{\mathrm{pool}}=B^{-1}\sum_{b=1}^BQ_{n,M}^{(b)}/w_{n,M}\), where \(w_{n,M}=\sum_{m=1}^M\gamma_m^2\) and \(N_{\mathrm{eff}}(n,M)=w_{n,M}^2/\sum_{m=1}^M\gamma_m^4\) is the effective number of weighted score observations.
Then
\[
    \left\|
    \widehat I_{n,M,B}^{\mathrm{pool}}
    -
    I(\theta_{n,0})
    \right\|_{\mathrm{op}}
    =
    \mathcal O_p
    \left[
        \sqrt{\frac{\log p}{B\,N_{\mathrm{eff}}(n,M)}}
        +
        \sqrt{w_{n,M}}
    \right],
\]
up to constants depending on \(S_K\) and \(L_K\). Consequently, if \(B\,N_{\mathrm{eff}}(n,M)\to\infty\) and \(w_{n,M}\to0\), then \(\widehat I_{n,M,B}^{\mathrm{pool}}\xrightarrow{p}I(\theta_{n,0})\).
In particular, for fixed \(M\), \(N_{\mathrm{eff}}(n,M)\asymp M\) and \(w_{n,M}\asymp M/n^2\); hence consistency follows if \(B\to\infty\).
\end{proposition}

\begin{proof}
Write \(\alpha_{n,m}=\gamma_m^2/w_{n,M}\), so that \(\sum_{m=1}^M\alpha_{n,m}=1\). Then \(\widehat I_{n,M,B}^{\mathrm{pool}}-I(\theta_{n,0})=R_{n,M,B}+B_{n,M,B}\), where \(R_{n,M,B}=B^{-1}\sum_{b=1}^B\sum_{m=1}^M\alpha_{n,m}[s_{n,m}^{(b)}\{s_{n,m}^{(b)}\}^\top-I(\theta_{n,m-1}^{(b)})]\) is the martingale fluctuation term and \(B_{n,M,B}=B^{-1}\sum_{b=1}^B\sum_{m=1}^M\alpha_{n,m}\{I(\theta_{n,m-1}^{(b)})-I(\theta_{n,0})\}\) is the path-drift bias.

For \(R_{n,M,B}\), the summands are conditionally mean-zero self-adjoint matrix martingale differences. Since \(\|s_{n,m}^{(b)}\|\le S_K\), their operator norms are bounded by a constant multiple of \(S_K^2\alpha_{n,m}/B\), and the matrix variance proxy is \(S_K^4B^{-1}\sum_{m=1}^M\alpha_{n,m}^2=S_K^4/[B\,N_{\mathrm{eff}}(n,M)]\). A standard matrix Freedman inequality therefore gives \(\|R_{n,M,B}\|_{\mathrm{op}}=\mathcal O_p[S_K^2\sqrt{\log p/\{B\,N_{\mathrm{eff}}(n,M)\}}]\), with the usual lower-order linear term absorbed in the rate.

For \(B_{n,M,B}\), Lipschitz continuity gives \(\|B_{n,M,B}\|_{\mathrm{op}}\le L_KB^{-1}\sum_{b=1}^B\sum_{m=1}^M\alpha_{n,m}\|\theta_{n,m-1}^{(b)}-\theta_{n,0}\|\). For each chain, \(\theta_{n,m-1}^{(b)}-\theta_{n,0}=\sum_{j=1}^{m-1}\gamma_js_{n,j}^{(b)}\), and martingale orthogonality yields \(\mathbb E\|\theta_{n,m-1}^{(b)}-\theta_{n,0}\|^2\le S_K^2\sum_{j=1}^{m-1}\gamma_j^2\le S_K^2w_{n,M}\). Hence \(\|B_{n,M,B}\|_{\mathrm{op}}=\mathcal O_p(L_KS_K\sqrt{w_{n,M}})\).
Combining the two bounds proves the result.
\end{proof}

\subsection{Proofs of Section 3.1}
\label{app:proof_thm1}

\paragraph{Positive definiteness of \(Q_{n,\infty}\) for vMF.}
It remains to show that the limiting quadratic variation is nonsingular. 
Fix a proper linear subspace \(V\subsetneq\mathbb R^p\). For any finite natural parameter \(c\), \(\mathbb P_c\{s(X,c)\in V\}=\mathbb P_c\{X-\mathbb E_c(X)\in V\}=\mathbb P_c\{X\in\mathbb E_c(X)+V\}=0\).
Indeed, \(\mathbb E_c(X)+V\) is a proper affine subspace of \(\mathbb R^p\), so
its intersection with \(S^{p-1}\) has surface measure zero; the vMF density is
strictly positive with respect to surface measure.

Let \(V_m=\operatorname{span}(s_{n,1},\ldots,s_{n,m})\), with \(V_0=\{0\}\).
If \(V_{m-1}\neq \mathbb R^p\), then conditional on \(\mathcal F_{n,m-1}\),
\(V_{m-1}\) is a fixed proper subspace and \(c_{n,m-1}\) is finite. Hence \(\mathbb P(s_{n,m}\in V_{m-1}\mid\mathcal F_{n,m-1})=0\). Thus, until the span reaches \(\mathbb R^p\), each new score increases the dimension by one almost surely; consequently \(\operatorname{span}(s_{n,1},\ldots,s_{n,p})=\mathbb R^p\) a.s.
Therefore, for any nonzero \(a\in\mathbb R^p\), at least one of \(a^\top s_{n,1},\ldots,a^\top s_{n,p}\) is nonzero, and hence \(a^\top Q_{n,\infty}a=\sum_{m=1}^{\infty}\gamma_m^2(a^\top s_{n,m})^2\ge\sum_{m=1}^p\gamma_m^2(a^\top s_{n,m})^2>0\) a.s.
Thus \(Q_{n,\infty}\) is positive definite almost surely.

\begin{proof}[Proof (Theorem \ref{thm:as_convergence})]
Fix \(n\) and work conditionally on the observed data \(X_{1:n}\). Define \(Y_{n,m}=\gamma_ms_{n,m}\) and \(D_{n,M}=\sum_{m=1}^MY_{n,m}\). By Assumption \ref{assum:martingale_convergence}(i), \(\mathbb E(Y_{n,m}\mid\mathcal F_{n,m-1})=0\), so \((D_{n,M},\mathcal F_{n,M})_{M\ge1}\) is a martingale.

For \(L>M\), martingale orthogonality gives
\[
\begin{aligned}
    \mathbb E
    \left[
        \|D_{n,L}-D_{n,M}\|^2
        \mid \mathcal F_{n,0}
    \right]
    &=
    \mathbb E
    \left[
        \sum_{m=M+1}^L
        \|Y_{n,m}\|^2
        \mid \mathcal F_{n,0}
    \right]   \\
    &=
    \mathbb E
    \left[
        \sum_{m=M+1}^L
        \gamma_{m}^2
        \mathbb E\{\|s_{n,m}\|^2\mid \mathcal F_{n,m-1}\}
        \mid \mathcal F_{n,0}
    \right].
\end{aligned}
\]
By Assumption \ref{assum:martingale_convergence}(ii) and conditional monotone convergence, the right-hand side tends to zero as \(M,L\to\infty\). Hence \((D_{n,M})_{M\ge1}\) is Cauchy in conditional \(L^2\), and therefore converges in conditional \(L^2\) to a finite random variable \(D_{n,\infty}\). The martingale convergence theorem also yields \(D_{n,M}\to D_{n,\infty}\) a.s.; since \(\theta_{n,M}=\theta_{n,0}+D_{n,M}\), we obtain \(\theta_{n,M}\to\theta_{n,\infty}:=\theta_{n,0}+D_{n,\infty}\) a.s.

Next, by Assumption \ref{assum:martingale_convergence}(iii), \(Q_{n,M}\to Q_{n,\infty}\) a.s., where \(Q_{n,\infty}\) is finite and positive definite. Since \(w_{n,M}=\sum_{m=1}^M\gamma_m^2\to w_{n,\infty}=\sum_{m=1}^\infty\gamma_m^2\in(0,\infty)\), it follows that \(\widehat I_{n,M}^{\mathrm{path}}=Q_{n,M}/w_{n,M}\to\widehat I_{n,\infty}^{\mathrm{path}}=Q_{n,\infty}/w_{n,\infty}\) a.s. The limiting matrix \(\widehat I_{n,\infty}^{\mathrm{path}}\) is positive definite; continuity of matrix inversion therefore gives \((\widehat I_{n,M}^{\mathrm{path}})^{-1}\to(\widehat I_{n,\infty}^{\mathrm{path}})^{-1}\) a.s.
Combining this convergence with \(D_{n,M}\to D_{n,\infty}\) gives
\[
\begin{aligned}
    \theta_{n,M}^{\mathrm{pc}}
    &=
    \theta_{n,0}
    +
    \left(\widehat I_{n,M}^{\mathrm{path}}\right)^{-1}D_{n,M}  \\
    &\to
    \theta_{n,0}
    +
    \left(\widehat I_{n,\infty}^{\mathrm{path}}\right)^{-1}D_{n,\infty}
    =
    \theta_{n,\infty}^{\mathrm{pc}}
    \qquad \text{a.s.}
\end{aligned}
\]
This proves the theorem.
\end{proof}

\subsection{Proof of Theorem \ref{thm:hybrid_bound}}
\label{app:proof_thm2}

\begin{proof}
We give the proof for one predictive chain. The same argument applies conditionally on any larger sigma-field containing the first \(M\) steps and the calibration estimator; this covers both the pathwise and pooled calibration choices in Algorithm \ref{alg:vmf_mps}.

Let \(H_M=I(\theta_{n,M})\) and \(\widehat I=\widehat I_{n,M}\). Under the conditional localization assumption in the theorem, the predictive continuation remains in \(K\) almost surely and \(\|\widehat I-H_M\|_{\mathrm{op}}\le\varepsilon_{n,M}\), \(\lambda_{\min}(\widehat I)\ge\lambda_K/2\), and \(\|\widehat I^{-1}\|_{\mathrm{op}}\le2/\lambda_K\).

Define \(f(y)=\varphi[\theta_{n,0}+\widehat I^{-1}\{D_{n,M}+y\}]\) for \(y\in\mathbb R^p\). Then \(\varphi(\theta_{n,M}^{\mathrm{tail,pc}})=f(T_{n,M})\), \(\varphi(\theta_{n,M}^{\mathrm{hyb,pc}})=f(Z_{n,M})\), \(T_{n,M}=\sum_{m=M+1}^\infty\gamma_ms_{n,m}\), and \(Z_{n,M}\sim N_p(0,r_{n,M}\widehat I)\). Moreover, \(\|D^2f\|_\infty\le4\lambda_K^{-2}\|D^2\varphi\|_\infty\) and \(\|D^3f\|_\infty\le8\lambda_K^{-3}\|D^3\varphi\|_\infty\).

Let \(G_M\sim N_p(0,r_{n,M}H_M)\).
By the triangle inequality,
\[
\begin{aligned}
    &
    \left|
    \mathbb E\{\varphi(\theta_{n,M}^{\mathrm{tail,pc}})\mid \mathcal F_{n,M}\}
    -
    \mathbb E\{\varphi(\theta_{n,M}^{\mathrm{hyb,pc}})\mid \mathcal F_{n,M}\}
    \right|  \\
    &\qquad
    \le
    \left|
    \mathbb E\{f(T_{n,M})\mid \mathcal F_{n,M}\}
    -
    \mathbb E\{f(G_M)\mid \mathcal F_{n,M}\}
    \right|  \\
    &\qquad\quad
    +
    \left|
    \mathbb E\{f(G_M)\mid \mathcal F_{n,M}\}
    -
    \mathbb E\{f(Z_{n,M})\mid \mathcal F_{n,M}\}
    \right|.
\end{aligned}
\]

The first term is controlled by Lemma \ref{lem:tail_gaussian}:
\[
\begin{aligned}
    \left|
    \mathbb E\{f(T_{n,M})\mid \mathcal F_{n,M}\}
    -
    \mathbb E\{f(G_M)\mid \mathcal F_{n,M}\}
    \right|
    \le
    C_p
    \left\{
        \|D^2 f\|_\infty L_KS_K r_{n,M}^{3/2}
        +
        \|D^3 f\|_\infty S_K^3 A_{n,M}
    \right\}.
\end{aligned}
\]
Substituting the derivative bounds for \(f\) gives \(C_p\{\lambda_K^{-2}\|D^2\varphi\|_\infty L_KS_Kr_{n,M}^{3/2}+\lambda_K^{-3}\|D^3\varphi\|_\infty S_K^3A_{n,M}\}\).

For the second term, use Lemma \ref{lem:gaussian_cov_perturb}. The covariance matrices are \(r_{n,M}H_M\) and \(r_{n,M}\widehat I\), so \(\|r_{n,M}H_M-r_{n,M}\widehat I\|_{\mathrm{op}}\le r_{n,M}\varepsilon_{n,M}\).
Therefore,
\[
\begin{aligned}
    \left|
    \mathbb E\{f(G_M)\mid \mathcal F_{n,M}\}
    -
    \mathbb E\{f(Z_{n,M})\mid \mathcal F_{n,M}\}
    \right|
    &\le
    C_p
    \|D^2 f\|_\infty
    r_{n,M}\varepsilon_{n,M}  \\
    &\le
    C_p
    \frac{\|D^2\varphi\|_\infty}{\lambda_K^2}
    r_{n,M}\varepsilon_{n,M}.
\end{aligned}
\]
Combining the two bounds yields
\[
\begin{aligned}
    &
    \left|
    \mathbb E\{\varphi(\theta_{n,M}^{\mathrm{tail,pc}})\mid \mathcal F_{n,M}\}
    -
    \mathbb E\{\varphi(\theta_{n,M}^{\mathrm{hyb,pc}})\mid \mathcal F_{n,M}\}
    \right|   \\
    &\qquad\le
    C_p
    \left[
    \frac{\|D^2\varphi\|_\infty}{\lambda_K^2}
    \left\{
        L_KS_K r_{n,M}^{3/2}
        +
        r_{n,M}\varepsilon_{n,M}
    \right\}
    +
    \frac{\|D^3\varphi\|_\infty}{\lambda_K^3}
    S_K^3 A_{n,M}
    \right],
\end{aligned}
\]
which is the claimed bound.
\end{proof}

\subsection{Proof of Theorem \ref{thm:bvm}}
\label{app:proof_thm3}

\begin{proof}
We prove the stated smooth-test-function bound. The weak convergence statement follows by standard approximation arguments.

Let \(M=M_n\), \(w_n=w_{n,M_n}\), \(r_n=r_{n,M_n}\), and
\(\widehat I=\widehat I_n\). The hybrid draw and its scaled decomposition are
\[
\begin{aligned}
    \theta_n^{\mathrm{hyb,pc}}&=\theta_{n,0}+\widehat I^{-1}\{D_{n,M}+Z_{n,M}\},
    & Z_{n,M}\mid(\mathcal F_{n,M}\vee\mathcal G_n)&\sim\mathcal N_p(0,r_n\widehat I),\\
    Y_n:=\sqrt n\{\theta_n^{\mathrm{hyb,pc}}-\theta_{n,0}\}&=R_n+G_n,
    & (R_n,G_n)&=\sqrt n\,\widehat I^{-1}(D_{n,M},Z_{n,M}),\\
    Z^*&\sim\mathcal N_p\{0,I(\theta^*)^{-1}\}.&&
\end{aligned}
\]

By Assumption \ref{assum:mle_bvm}, with probability tending to one,
\(\theta_{n,0}\) and the explicit predictive path remain in the compact neighborhood
\(K\), and
\[
    \lambda_{\min}(\widehat I)\ge\tfrac12\lambda_{\min}\{I(\theta^*)\},
    \qquad \|\widehat I^{-1}\|_{\mathrm{op}}\le C
\]
for a finite constant \(C\). The bounds below are established on this event; the
complement has probability tending to zero.

First consider the explicitly simulated part \(R_n\). By martingale orthogonality and boundedness of the score on \(K\),
\[
\begin{aligned}
    \mathbb E
    \left(
        \|D_{n,M}\|^2
        \mid \mathcal G_n
    \right)
    &=
    \mathbb E
    \left[
        \left\|
            \sum_{m=1}^{M}
            \gamma_{m}s_{n,m}
        \right\|^2
        \mid \mathcal G_n
    \right]
\le
    C
    \sum_{m=1}^{M}\gamma_{m}^2
    =
    Cw_n.
\end{aligned}
\]
Therefore,
\[
\begin{aligned}
    \mathbb E
    \left(
        \|R_n\|
        \mid \mathcal G_n
    \right)
    &\le
    \sqrt n\,
    \|\widehat I^{-1}\|_{\mathrm{op}}\,
    \mathbb E
    \left(
        \|D_{n,M}\|
        \mid \mathcal G_n
    \right)
    \le
    C\sqrt{n w_n}.
\end{aligned}
\]
Since \(w_n=\sum_{m=1}^{M_n}(n+m)^{-2}\le M_n/n^2\), we obtain \(\mathbb E(\|R_n\|\mid\mathcal G_n)\le C\sqrt{M_n/n}\).
Thus, for every \(\psi\in C_b^2(\mathbb R^p)\),
\[
\begin{aligned}
    &
    \left|
    \mathbb E
    \left\{
        \psi(R_n+G_n)
        \mid \mathcal G_n
    \right\}
    -
    \mathbb E
    \left\{
        \psi(G_n)
        \mid \mathcal G_n
    \right\}
    \right|
    \le
    \|D\psi\|_\infty\,
    \mathbb E
    \left(
        \|R_n\|
        \mid \mathcal G_n
    \right)
    =
    \mathcal O_p
    \left(
        \sqrt{\frac{M_n}{n}}
    \right).
\end{aligned}
\]

It remains to compare \(G_n\) with \(Z^*\). Conditional on \(\mathcal G_n\), \(G_n\sim\mathcal N_p(0,\Sigma_n)\), where \(\Sigma_n=nr_n\widehat I^{-1}\).
By the Gaussian covariance perturbation lemma,
\[
\begin{aligned}
    &
    \left|
    \mathbb E
    \left\{
        \psi(G_n)
        \mid \mathcal G_n
    \right\}
    -
    \mathbb E\{\psi(Z^*)\}
    \right|
    \le
    C_p\|D^2\psi\|_\infty
    \left\|
        \Sigma_n-I(\theta^*)^{-1}
    \right\|_{\mathrm{op}}.
\end{aligned}
\]
We now bound this covariance difference. Since \(w_{n,\infty}=\sum_{m=1}^{\infty}(n+m)^{-2}=n^{-1}+\mathcal O(n^{-2})\), \(r_n=w_{n,\infty}-w_n\), and \(nw_n\le M_n/n\), we have \(nr_n=1+\mathcal O(n^{-1})-nw_n\) and hence \(|nr_n-1|=\mathcal O(n^{-1}+M_n/n)\).
It follows that
\[
\begin{aligned}
    \left\|
        \Sigma_n-I(\theta^*)^{-1}
    \right\|_{\mathrm{op}}
    &=
    \left\|
        nr_n\widehat I^{-1}
        -
        I(\theta^*)^{-1}
    \right\|_{\mathrm{op}}
    \\
    &\le
    |nr_n-1|\,
    \|\widehat I^{-1}\|_{\mathrm{op}}
    +
    \left\|
        \widehat I^{-1}
        -
        I(\theta^*)^{-1}
    \right\|_{\mathrm{op}}.
\end{aligned}
\]
The first term is \(\mathcal O_p(n^{-1}+M_n/n)\). For the second term, use \(A^{-1}-B^{-1}=A^{-1}(B-A)B^{-1}\)
with \(A=\widehat I\) and \(B=I(\theta^*)\). On the high-probability event where both matrices are uniformly nonsingular,
\[
\begin{aligned}
    \left\|
        \widehat I^{-1}
        -
        I(\theta^*)^{-1}
    \right\|_{\mathrm{op}}
    &\le
    C
    \left\|
        \widehat I-I(\theta^*)
    \right\|_{\mathrm{op}}
\le
    C
    \left[
        \left\|
            \widehat I-I(\theta_{n,0})
        \right\|_{\mathrm{op}}
        +
        \left\|
            I(\theta_{n,0})-I(\theta^*)
        \right\|_{\mathrm{op}}
    \right].
\end{aligned}
\]
The first term is \(\Delta_{I,n}\). The second term is \(\mathcal O_p(n^{-1/2})\), by local Lipschitz continuity of \(I(\cdot)\) and the asymptotic normality of \(\theta_{n,0}\). Hence \(\|\widehat I^{-1}-I(\theta^*)^{-1}\|_{\mathrm{op}}=\mathcal O_p(\Delta_{I,n}+n^{-1/2})\), and combining the preceding bounds gives \(\|\Sigma_n-I(\theta^*)^{-1}\|_{\mathrm{op}}=\mathcal O_p(\Delta_{I,n}+n^{-1/2}+M_n/n)\). Since \(M_n/n\to0\), \(M_n/n\le\sqrt{M_n/n}\) for all sufficiently large \(n\). Therefore,
\[
\begin{aligned}
    &
    \left|
    \mathbb E
    \left\{
        \psi(G_n)
        \mid \mathcal G_n
    \right\}
    -
    \mathbb E\{\psi(Z^*)\}
    \right|
    =
    \mathcal O_p
    \left(
        \Delta_{I,n}
        +
        n^{-1/2}
        +
        \sqrt{\frac{M_n}{n}}
    \right).
\end{aligned}
\]
Combining this bound with the preceding bound for \(R_n\) yields
\[
\begin{aligned}
    &
    \left|
    \mathbb E
    \left[
        \psi
        \left\{
            \sqrt n
            (\theta_n^{\mathrm{hyb,pc}}-\theta_{n,0})
        \right\}
        \mid \mathcal G_n
    \right]
    -
    \mathbb E\{\psi(Z^*)\}
    \right|
    =
    \mathcal O_p
    \left(
        \Delta_{I,n}
        +
        \sqrt{\frac{M_n}{n}}
        +
        n^{-1/2}
    \right).
\end{aligned}
\]
Since \(\Delta_{I,n}=o_p(1)\) and \(M_n/n\to0\), the right-hand side converges to zero in probability. This proves the conditional convergence given \(\mathcal G_n\).

If \(\widehat I_n\) is data-measurable, then \(\mathcal G_n=\sigma(X_{1:n})\). If \(\widehat I_n\) is generated from an independent calibration ensemble, the conditional-on-data version follows from the tower property: the conditional discrepancy given \(X_{1:n}\) is bounded above by the conditional expectation, given \(X_{1:n}\), of the discrepancy just shown to converge to zero in probability. Boundedness of \(\psi\) then gives the result. This completes the proof.
\end{proof}

\subsection{Proof of Corollary \ref{cor:vmf_bound}}
\label{app:proof_vmf_local}

\begin{proof}
For the vMF model in natural parameter \(c\), the score is \(s(x,c)=x-A_p(\|c\|)c/\|c\|\), with the usual convention at \(c=0\). Since \(\|x\|=1\) and \(0\le A_p(\kappa)<1\), \(\|s(x,c)\|\le1+A_p(\|c\|)\le2\).
Thus the score is globally bounded.

For \(c\ne0\), write \(\kappa=\|c\|\) and \(u=c/\kappa\). The Fisher information is \(I(c)=\{A_p(\kappa)/\kappa\}(I_p-uu^\top)+A_p'(\kappa)uu^\top\), with eigenvalues \(A_p(\kappa)/\kappa\) on the tangent space and \(A_p'(\kappa)\) radially. At \(\kappa=0\), the standard expansion \(A_p(\kappa)=\kappa/p+\mathcal O(\kappa^3)\) implies \(A_p(\kappa)/\kappa\to1/p\) and \(A_p'(\kappa)\to1/p\).
Therefore \(I(c)\to p^{-1}I_p\) as \(c\to0\).

Fix \(R<\infty\) and let \(K_R=\{c\in\mathbb R^p:\|c\|\le R\}\). The functions \(A_p(\kappa)/\kappa\) and \(A_p'(\kappa)\), interpreted continuously at \(\kappa=0\), are continuous and strictly positive on \([0,R]\); hence \(\lambda_R:=\inf_{0\le\kappa\le R}\min\{A_p(\kappa)/\kappa,A_p'(\kappa)\}>0\).
This proves uniform nonsingularity on \(K_R\).

It remains to note Lipschitz continuity. Write \(I(c)=a(\kappa)I_p+\{b(\kappa)-a(\kappa)\}uu^\top\), where \(a(\kappa)=A_p(\kappa)/\kappa\) and \(b(\kappa)=A_p'(\kappa)\). The expansions \(a(\kappa)=p^{-1}+\mathcal O(\kappa^2)\) and \(b(\kappa)=p^{-1}+\mathcal O(\kappa^2)\) show that \(b(\kappa)-a(\kappa)=\mathcal O(\kappa^2)\) as \(\kappa\downarrow0\).
Thus the apparent singularity in \(uu^\top=cc^\top/\|c\|^2\) is removable, and \(I(c)\) extends continuously, indeed locally Lipschitzly, to \(c=0\). Away from zero, \(I(c)\) is smooth. Since \(K_R\) is compact, \(I(c)\) is Lipschitz on \(K_R\). Therefore Assumption \ref{assum:hybrid_regularity} holds on every compact \(K_R\), completing the proof.
\end{proof}

\subsection{MCMC benchmarks}
\label{app:mcmc}

\paragraph{Circular model.}
For angular observations $\omega_i$ with $X_i=(\cos\omega_i,\sin\omega_i)^\top$, the von Mises density is $p(\omega\mid\phi,\kappa)\propto I_0(\kappa)^{-1}\exp\{\kappa\cos(\omega-\phi)\}$. Under the conjugate prior $\pi(\phi,\kappa)\propto I_0(\kappa)^{-c_0}\exp\{\kappa\mathcal R_0\cos(\phi-\phi_0)\}$ with $c_0,\mathcal R_0\ge0$ and $\phi_0\in[0,2\pi)$, the joint posterior is
\[
    \pi(\phi,\kappa\mid X_{1:n})
    \propto
    I_0(\kappa)^{-(c_0+n)}\exp\{\kappa\mathcal R_n\cos(\phi-\bar\omega_n)\},
\]
where $\mathcal R_n=\|\mathcal R_0\mu_{\phi_0}+n\bar X_n\|$ and $\bar\omega_n$ are the resultant length and mean direction of the combined prior and data, with $\mu_{\phi_0}=(\cos\phi_0,\sin\phi_0)^\top$. 

Posterior sampling is conducted via a two-step Gibbs sampler. Given the state $(\phi, \kappa)$:
\begin{enumerate}
    \item \textbf{Update $\phi$:} Draw $\phi \mid \kappa, X_{1:n}$ from a von Mises distribution with mean direction $\bar{\omega}_n$ and concentration $\kappa \mathcal{R}_n$.
    \item \textbf{Update $\kappa$:} Draw $\kappa \mid \phi, X_{1:n}$ from the full conditional $\pi(\kappa \mid \phi, X_{1:n}) \propto I_0(\kappa)^{-\eta} \exp\{-\eta \beta_0 \kappa\}$, where $\eta = c_0 + n$ and $\beta_0 = -\frac{\mathcal{R}_n}{\eta} \cos(\phi - \bar{\omega}_n)$, utilizing the algorithm of \cite{forbes2015fast}.
\end{enumerate}

 We set $c_0=\mathcal R_0=\phi_0=0$ throughout.

\paragraph{Spherical model.}
The circular construction above uses the closed form of the von Mises normalizing constant and does not extend to $p=3$. We therefore follow \citet{straub2017nonparametric}, whose joint prior
\[
    p(\mu,\kappa\mid\mu_0,a,b)
    \;\propto\;
    \Bigl(\frac{\kappa}{\sinh\kappa}\Bigr)^{a}
    \exp\bigl(b\,\kappa\,\mu^\top\mu_0\bigr)
\]
is specified only up to proportionality and is proper for $0<b<a$. Given observations $X_{1:n}$, the posterior belongs to the same family, with parameters
\[
    a_N = a+n,
    \qquad
    b_N\mu_N = b\,\mu_0 + \sum\nolimits_{i=1}^{n} X_i .
\]
The full conditional for $\mu$ is von Mises--Fisher with mean direction $\mu_N$ and concentration $b_N\kappa$, and is sampled directly. The full conditional for $\kappa$ is unimodal but has no tractable inverse cumulative distribution function, and is sampled by the slice sampler of that reference, with the slice boundaries located by Newton iteration.

As in the circular case, we use the improper limiting prior obtained as $a,b\to0^{+}$, under which $\mu_0$ does not enter the posterior and the updates reduce to
\[
    a_N = n,
    \qquad
    b_N = n R_n,
    \qquad
    \mu_N = \bar X_n / R_n ,
\]
with $\bar X_n$ and $R_n=\|\bar X_n\|$ as in the initialization. Whenever $0<R_n<1$ these satisfy $0<b_N<a_N$, so the posterior lies in the conjugate family with valid parameters and provides a proper reference posterior.

\subsection{More results for Section~\ref{sec:single_sim}}
\label{app:single_sim_n=100}


Section~\ref{sec:single_sim} reports the two extreme sample sizes, $n=10$ and $n=500$. The intermediate case $n=100$ is shown here for completeness. At this sample size the truncation factor $\{w_{n,M}/w_{n,\infty}\}^{1/2}$ equals $0.953$ at simulation depth $M=1000$, so MPS-R retains about $95\%$ of the posterior scale and the variance deficit is an order of magnitude smaller than at $n=500$.

\begin{figure}[H]
    \centering
    \includegraphics[width=0.9\textwidth]{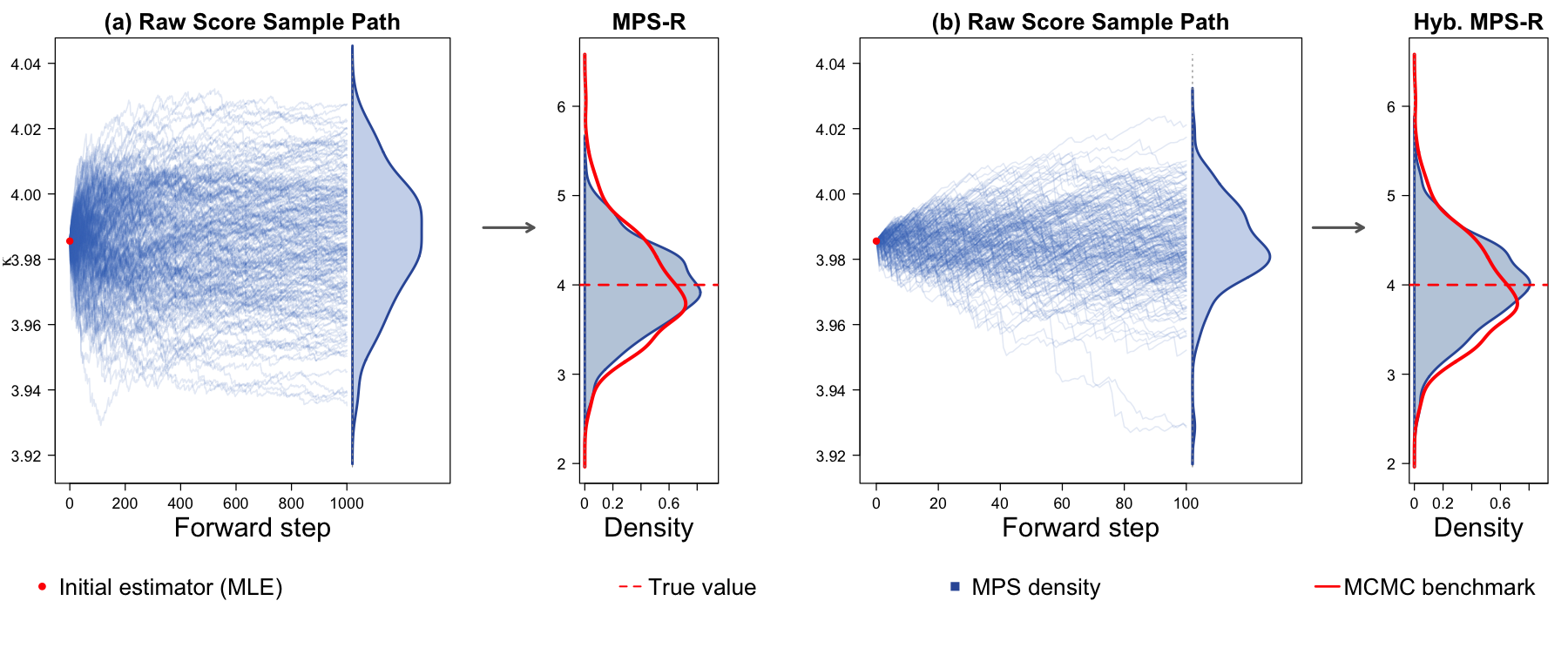}
    \captionsetup{width=0.90\linewidth, font=small}
    \caption{Marginal posterior density of $\kappa$: (a) $n=100$, MPS-R, $M=1000$, (b) $n=100$, Hybrid MPS-R, $M=100$.}
    \label{fig:MPS_n100}
\end{figure}

The marginal posterior of $\kappa$ is close to Gaussian at this sample size, and the right skewness present at $n=10$ has largely disappeared. MPS-R is not visibly under-dispersed in panel (a) of Figure~\ref{fig:MPS_n100}, in contrast to the $n=500$ case in Figure~\ref{fig:MPS_n500}, and Hybrid MPS-R reproduces the same density in panel (b) from a path of length $M=100$. The residual difference between the two panels is consistent with the interval lengths reported in Table~\ref{tab:sim_results_p2} at $n=100$, where the ratio of the MPS-R to the Hybrid MPS-R lengths is $0.952$ for $\kappa$ and $0.958$ for $\phi$. Both densities remain close to the MCMC benchmark, whose intervals for $\kappa$ are somewhat wider at this sample size, as discussed in Section~\ref{sec:coverage}.


\end{document}